\documentclass[12pt]{article}
\usepackage{amsmath,amssymb,amsthm,mathtools}
\usepackage{comment}
\usepackage[margin=1in]{geometry}
\usepackage{hyperref}
\usepackage{graphicx}
\usepackage{subcaption}
\usepackage{pifont}
\usepackage{arydshln}
\usepackage{placeins}

\newtheorem{theorem}{Theorem}[section]

\newtheorem{corollary}[theorem]{Corollary}
\newtheorem{proposition}[theorem]{Proposition}

\newtheorem{remark}{Remark}[section]

\newcommand{\R}{\mathbb{R}}
\newcommand{\Z}{\mathbb{Z}}
\newcommand{\C}{\mathbb{C}}
\newcommand{\vb}[1]{\mathbf{#1}}
\newcommand{\ket}[1]{|#1\rangle}
\newcommand{\bra}[1]{\langle#1|}
\newcommand{\braket}[2]{\langle#1|#2\rangle}
\newcommand{\ketbra}[2]{|#1\rangle\langle#2|}
\newcommand{\Tr}{\operatorname{Tr}}
\newcommand{\EE}{\mathbb{E}}
\newcommand{\dW}{\mathrm{d}W}
\newcommand{\dd}{\mathrm{d}}
\newcommand{\grad}{\nabla}
\newcommand{\divg}{\nabla\cdot}

\newcommand{\lapl}{\nabla^2}

\title{Why Does Classical Turbulence Obey an Area Law?}
\author{Wael Itani$^{a,b}$\footnote{\href{mailto:wi07@aub.edu.lb}{wi07@aub.edu.lb}}\\[6pt]
$^{a}$Maroun Semaan Faculty of Engineering \& Architecture,\\[3pt]
$^{b}$Center for Advanced Mathematical Sciences,\\
American University of Beirut, Beirut 1107 2020, Lebanon}
\date{\today}

\begin{document}
\maketitle

\begin{abstract}
In incompressible flow the viscous force is solenoidal, whereas the Madelung transform of a spinless Schr\"odinger equation produces only gradient forces. The two are orthogonal, so viscosity cannot arise from Hamiltonian quantum mechanics alone; an open quantum treatment is required. Reducing the $N$-body density matrix to its one-body component and closing the dynamics via Born-Markov yields Lindblad jump operators with $k^2$ scattering rates, which we unravel via quantum state diffusion (QSD) into a norm-preserving stochastic nonlinear Schr\"odinger equation. Dissipation and stochastic forcing are not separate ingredients: both come from the same Lindblad operators, and their amplitudes are locked by the QSD structure. The Madelung transform of this equation, under incompressibility, gives a stochastic Navier-Stokes equation whose viscosity is set by the mean free path and whose noise correlator satisfies the fluctuation-dissipation relation by construction, in agreement with the Landau-Lifshitz framework. The recovery is conditional: the viscous identification holds at the ensemble level via the vortex decomposition of the velocity field; the single-trajectory identification remains open. The zeros of the wavefunction carry quantised circulation; their codimension-2 topology yields the Migdal area law for circulation statistics~\cite{migdal1994} under a Poisson assumption, here through a different mechanism than the loop-functional saddle point and verified numerically even in the quantum regime where the de~Broglie length exceeds the Kolmogorov scale.
\end{abstract}

\newpage
\tableofcontents
\newpage

\section{Introduction}\label{sec:introduction}

The hydrodynamic interpretation of quantum mechanics~\cite{madelung1927} recasts the Schr\"odinger equation as a continuity equation for the probability density and a compressible Euler-type momentum equation augmented by the quantum potential. It has long served as a conceptual bridge between quantum and classical mechanics~\cite{bohm1952a,bohm1952b,sonego1991}, but its implications for classical fluid dynamics, and turbulence in particular, remain underexplored. We show that introducing dissipation into the Madelung framework via the Born-Markov approximation recovers the Navier-Stokes equations with a well-defined kinematic viscosity, under incompressibility and at the ensemble level via the vortex decomposition of the velocity field. The codimension-2 topology of the wavefunction zeros then yields the Migdal area law for circulation statistics~\cite{migdal1994,migdal2019,migdal2024,migdal2024b} via a topological mechanism distinct from the loop-functional saddle point on which the law was originally derived.

The Madelung transform $\psi = \sqrt{\rho}\,e^{iS/\hbar}$ identifies $\rho$ with density and $\grad S/m$ with velocity, yielding an inviscid, irrotational flow~\cite{madelung1927,wyatt2005,takabayasi1952,schonberg1954}. Previous attempts to add viscosity span stochastic mechanics~\cite{nelson1966,nelson1967,nelson1985} (placed on variational foundations~\cite{yasue1981,guerra1983}, with quantisation subtleties~\cite{wallstrom1989}), non-Hermitian Schr\"odinger equations~\cite{succi2023}, spinor formulations~\cite{meng2023,yang2024}, and quantum computational methods~\cite{yepez2001,gaitan2020,li2025,sanavio2024,itani2024qalb,itani2021toolbox,succi2023where,zhu2025}. Within the single-component (spinless) setting, every deterministic route fails: complex kinetic coefficients~\cite{succi2023} introduce spurious mass diffusion and break norm conservation; logarithmic nonlinearities~\cite{doebner1992} give uniform friction rather than scale-selective dissipation; and spinor models~\cite{meng2023,yang2024} step outside the single-particle Lindblad framework. The reason is structural (\S\ref{sec:gradient-obstruction}): for incompressible flow $\nu\lapl\vb{u}$ is solenoidal, while Madelung forces are gradients, and the two are orthogonal under the Helmholtz decomposition. Related work has examined the quantum potential as a Korteweg capillary stress~\cite{krishnaswami2020}, its connection to Fisher information and barotropic pressure~\cite{heifetz2015}, geometric reformulations of the Madelung transform~\cite{reddiger2017}, and vortex dynamics in the perfect-liquid Madelung equations~\cite{sorokin2001}. The present work resolves the obstruction within the spinless, position-basis framework by unravelling the Lindblad master equation via QSD, so that dissipation and stochastic forcing emerge jointly from the same operators.

The construction proceeds from quantum to classical. Starting from the $N$-body Schr\"odinger equation, we trace out $N-1$ particles to obtain the one-body density matrix, derive Lindblad jump operators with $k^2$ rates via Born-Markov closure and Fermi's golden rule, unravel via QSD into a stochastic NLS whose trajectories $\psi(\vb{x},t)$ are scalar fields on $\mathbb{R}^3$, and apply the Madelung transform to conditionally recover the Navier-Stokes equations. To the extent that this recovery holds, standard turbulence results (the K\'arm\'an-Howarth equation, the four-fifths law, the $k^{-5/3}$ spectrum~\cite{karman1938,kolmogorov1941a,kolmogorov1941b}) follow as corollaries of Navier-Stokes, not as independent contributions of this work.

The defining feature of the construction is that dissipation enters through stochastic forcing: the Lindblad operators that generate viscous damping also generate the noise, and the two cannot be tuned independently. The fluctuation-dissipation relation is therefore a structural consequence of the QSD unravelling rather than an imposed condition, and the wavefunction representation is preserved throughout. The construction may inform future closure modelling, although no closure model is developed here~\cite{batchelor1953,monin1975}. The need for stochastic forcing in reduced turbulence models has been argued independently~\cite{freitas2025}: deterministic closures fail to capture uncertainty growth in shell models at high Reynolds number.

Section~\ref{sec:methodology} develops the framework: the gradient-solenoidal obstruction, the Lindblad derivation, QSD unravelling, and Navier-Stokes recovery. Section~\ref{sec:results} covers the classical limit and the area law for circulation. Section~\ref{sec:discussion} addresses robustness, sensitivity, and open problems. Section~\ref{sec:conclusion} summarises what is proven, assumed, and open. Numerical methods and convergence are deferred to Appendix~\ref{app:numerical}.

\subsection{Background}\label{sec:preliminaries}

We recall the incompressible Navier-Stokes equations~\cite{batchelor1967,chorin1993,temam2001,pope2000}:
\begin{align}
\frac{\partial \vb{u}}{\partial t} + (\vb{u} \cdot \grad)\vb{u} &= -\frac{1}{\rho_0}\grad p + \nu \lapl \vb{u}, \label{eq:ns-intro} \\
\divg \vb{u} &= 0, \label{eq:incomp-intro}
\end{align}
with Taylor-scale Reynolds number $Re_\lambda = u_{\mathrm{rms}}\,\lambda/\nu$, where the Taylor microscale is
\begin{equation}\label{eq:taylor-microscale}
\lambda = \sqrt{15\nu\,\langle u^2 \rangle/\varepsilon}
\end{equation}
(the statistical framework dates to~\cite{taylor1935}).

The Madelung transform~\cite{madelung1927,khesin2019} writes $\psi = \sqrt{\rho}\,e^{iS/\hbar}$, with density $\rho = |\psi|^2$, irrotational velocity $\vb{v} = \grad S/m$, and quantum potential $Q = -(\hbar^2/2m)\,\lapl\sqrt{\rho}/\sqrt{\rho}$. The obstruction to viscous dissipation under any Hamiltonian Schr\"odinger equation (self-adjoint $\hat{H}$, norm-preserving) is developed in \S\ref{sec:gradient-obstruction}; it extends to non-Hermitian modifications~\cite{succi2023}, deterministic nonlinear extensions~\cite{doebner1992}, and multi-component spinor formulations~\cite{meng2023,yang2024}, motivating the stochastic extension developed here.

Standard turbulence quantities referenced below are the Kolmogorov spectrum $E(k)=C_K\varepsilon^{2/3}k^{-5/3}$~\cite{sreenivasan1995}, the dissipation scale $\eta=(\nu^3/\varepsilon)^{1/4}$, the K\'arm\'an-Howarth equation~\cite{karman1938,batchelor1953} with its exact corollary $S_3(r)=-\tfrac{4}{5}\varepsilon r$~\cite{kolmogorov1941a,frisch1995,antonia2006}, the dissipative-anomaly conjecture of Onsager~\cite{onsager1945,onsager1949,eyink2006,eyink2024}, and the refined and anomalous-scaling extensions~\cite{kolmogorov1962,sheleveque1994,sreenivasan1997}.

\section{Methodology}\label{sec:methodology}

\subsection{Overview}\label{sec:framework}

\subsubsection{Derivation Strategy}

Vorticity in incompressible flows is carried by codimension-2 objects: point vortices in 2D and vortex filaments in 3D. As shown in \S\ref{sec:gradient-obstruction}, viscosity cannot arise from any Hamiltonian Schr\"odinger equation via the Madelung transform, which forces an open-system description through the Lindblad master equation.

We derive the $k^2$ scaling of the Lindblad jump operators from microscopic $N$-body scattering theory (\S\ref{sec:lindblad-k2}). Of the available unravellings into stochastic pure-state evolutions, we adopt quantum state diffusion (QSD): it yields a diffusive stochastic NLS with continuous trajectories that preserves the norm (\S\ref{sec:norm-preservation}), recovers the Lindblad master equation upon ensemble averaging (\S\ref{sec:ensemble-recovery}), and reproduces the incompressible Navier-Stokes equations under Madelung transform and incompressibility (\S\ref{sec:ns-recovery}). The momentum equation that results is a stochastic Navier-Stokes equation whose noise correlator satisfies the fluctuation-dissipation relation (\S\ref{sec:fdr}).

The remaining sections derive the area law for circulation statistics (\S\ref{sec:area-law}) from the codimension-2 topology of the wavefunction zeros and verify it numerically, including in the quantum regime $\mathrm{Kn}_q \gtrsim 1$.

Two physical ideas underlie the construction. First, viscosity emerges from decoherence: a Born-Markov trace over unresolved degrees of freedom in the $N$-body system produces dissipative dynamics, and a dimensional argument links the kinematic viscosity to the entanglement entropy production rate (\S\ref{sec:lindblad-k2}). Second, the zeros of the wavefunction are codimension-2 vortex cores, providing the topological framework for circulation quantisation and the area law.

\subsubsection{Natural Units}\label{sec:natural-units}

Three dimensional normalisations are used throughout. We set $\hbar = m = 1$, fixing the units of length and time and matching the simulation convention (Appendix~\ref{app:numerical}). In these units the circulation quantum is $\Gamma_0 = 2\pi$ and the de Broglie length is $\lambda_\mathrm{dB} = 1/v_\mathrm{rms}$. A third normalisation, the thermal energy density $k_BT/\rho$, appears in the fluctuation-dissipation ratio: matching the QSD value $S_{\vb{k}}/\gamma_{\vb{k}} = 2$ to the Landau-Lifshitz form $S_{\vb{k}} = 2\gamma_{\vb{k}}\cdot k_BT/\rho$ (\S\ref{sec:fdr}) requires $k_BT/\rho = 1$. Physical values are recovered by reinstating $\hbar_\mathrm{SI}$, $m_\mathrm{phys}$, and the equation of state. For water the dimensionless viscosity is $\nu_* = \nu_\mathrm{SI}\,m_\mathrm{phys}/\hbar_\mathrm{SI} \approx 283$ (Eq.~\eqref{eq:nu-star}).

\subsubsection{Dimensional Structure of the Correspondence}

The vortex-quantum correspondence hinges on the codimension of the vorticity-carrying structures (Table~\ref{tab:dimensional}). In 2D (point vortices) and 3D (vortex filaments), the carriers have codimension 2 and admit a geometric potential---the winding angle $\theta$ or the solid angle $\Omega$---with an associated topological invariant; the framework is complete in both cases. For codimension-3 point particles in 3D, the geometric bridge fails: no linking number, no area rule, and the Hamiltonian structure requires spinor wavefunctions. The analytical framework therefore applies to both 2D and 3D.

\begin{table}[h]
\centering
\begin{tabular}{lccc}
\hline
\textbf{Property} & \textbf{2D (Point Vortices)} & \textbf{3D (Filaments)} & \textbf{3D (Point Particles)} \\
\hline
Vorticity carrier & Points (codim-2) & Lines (codim-2) & Points (codim-3) \\
Geometric object & Winding angle $\theta$ & Solid angle $\Omega$ & Dipole (degenerate) \\
Topological invariant & Winding number & Linking number & None \\
Hamiltonian structure & Yes (Kirchhoff) & Yes (with constraints) & No (with stretching) \\
Wavefunction type & Spinless & Spinless & Requires spinor \\
Area rule & Holds (winding) & Holds (solid angle) & Breaks down \\
\hline
\end{tabular}
\caption{Dimensional dependence of the vortex-quantum correspondence. The framework is complete when vorticity carriers have codimension 2. The codimension-2 entries (2D points and 3D filaments) have been independently validated computationally~\cite{zhu2025}: vortex filaments in 3D turbulence can be extracted as the zero set of a complex scalar field using a variational quantum eigensolver, correctly capturing knotted and tangled filament topologies (septafoil, cinquefoil, and turbulent configurations).}
\label{tab:dimensional}
\end{table}

The constructive encoding of~\cite{zhu2025} demonstrates that a sufficiently regular divergence-free velocity field with quantised circulations~\cite{saffman1992} can be represented as the Madelung velocity of a complex scalar field $\psi$, with vortex filaments as the zero set $\mathcal{Z} = \{\vb{x} : \psi(\vb{x}) = 0\}$; the method reformulates the encoding as a Hermitian eigenvalue problem solved via a variational quantum eigensolver. The geometric connection between the wavefunction zeros and the classical vortex potentials (the winding angle $\theta$ in 2D and the solid angle $\Omega$ in 3D) is developed in \S\ref{sec:area-law}.

\subsection{The Need for Open Quantum Dynamics}\label{sec:gradient-obstruction}

Throughout this work, $\psi(\vb{x},t) = \braket{\vb{x}}{\psi(t)}$ is a single-component scalar wavefunction in the position basis---one complex number per point, no spin. It is \emph{not} the $N$-body wavefunction $\Psi(\vb{x}_1,\ldots,\vb{x}_N,t)$ on $3N$-dimensional configuration space, which cannot be Madelung-transformed to a fluid on $\mathbb{R}^3$; the construction of $\psi$ as a scalar field derived from the $N$-body dynamics is given in \S\ref{sec:reduction}. The obstruction below applies to any such $\psi$ regardless of origin, and its resolution turns on the \emph{stochastic} structure of the evolution rather than any modification of the Schr\"odinger operator on a single realisation. The position basis is forced by the Madelung transform: density $\rho = |\psi|^2$ and velocity $\vb{v} = \grad S/m$ are local in position space, and the kinetic operator $-\hbar^2\lapl/(2m)$ becomes the advection term $(\vb{v}\cdot\grad)\vb{v}$ plus quantum pressure $\grad Q/m$. Applying the Madelung transform to $i\hbar\partial_t\psi = -(\hbar^2/2m)\lapl\psi + V\psi$ yields the continuity equation $\partial_t\rho + \divg(\rho\vb{v}) = 0$ and the momentum equation
\begin{equation}\label{eq:madelung-hj}
\frac{\partial\vb{v}}{\partial t} + (\vb{v}\cdot\grad)\vb{v}
= -\frac{1}{m}\grad(V + Q),
\end{equation}
where $Q = -(\hbar^2/2m)\,\lapl\sqrt{\rho}/\sqrt{\rho}$ is the quantum pressure~\cite{madelung1927,bohm1952a}. Wherever $\psi \neq 0$ the velocity $\vb{v} = \grad S/m$ is irrotational and the right-hand side of~\eqref{eq:madelung-hj} is a pure gradient. For incompressible flow the viscous term $\nu\lapl\vb{v}$ is solenoidal, since $\divg(\nu\lapl\vb{v}) = \nu\lapl(\divg\vb{v}) = 0$. A field that is both a gradient and divergence-free must vanish: $\vb{f} = \grad\phi$ with $\divg\vb{f} = 0$ implies $\lapl\phi = 0$, whose only solution on a periodic or decaying domain is constant, so $\vb{f} = 0$. Therefore:

\emph{No conservative Schr\"odinger equation can produce $\nu\lapl\vb{v}$ via the Madelung transform in the smooth region where $\psi \neq 0$.}

At the zeros of $\psi$ the Madelung representation is singular: $S$ is multi-valued and $\vb{v}$ diverges as $\hbar/(mr_\perp)$, where $r_\perp$ is the distance to the zero set. In 2D the zeros are isolated points with vorticity $\vb{\omega} = (2\pi n\hbar/m)\,\delta^2(\vb{x} - \vb{x}_0)\,\hat{\vb{z}}$; in 3D they form filaments with $\vb{\omega} = (2\pi n\hbar/m)\,\delta^2(\vb{x}_\perp - \vb{x}_{0,\perp})\,\hat{\vb{t}}$, where $\hat{\vb{t}}$ is the unit tangent and $\delta^2$ acts in the transverse plane.

The resolution is to extend the conservative Schr\"odinger dynamics to an open quantum system, combining three ingredients: (i) deterministic mode-by-mode damping at rate $\nu k^2$ (derived from $N$-body scattering in \S\ref{sec:lindblad-k2}), which alone would violate norm conservation~\cite{succi2023} and introduce spurious mass diffusion~\cite{lighthill1963}; (ii) a nonlinear state-dependent correction that restores norm preservation trajectory by trajectory (Theorem~\ref{thm:norm}); and (iii) stochastic noise from the QSD unravelling that nucleates and annihilates vortex-antivortex pairs by transversality~\cite{guillemin1974}, populating the singular set $\mathcal{Z} = \{\vb{x} : \psi(\vb{x}) = 0\}$ whose collective dynamics carries the solenoidal vorticity that the orthogonality argument forbids in the smooth region. The explicit dynamics are derived in \S\ref{sec:qsd-derivation}--\ref{sec:ns-recovery}.

This result applies to wavefunctions in the position basis with no internal degrees of freedom, for which $\vb{v} = \grad S/m$ is necessarily irrotational away from zeros. Real fluid molecules are not spinless: they carry electronic, rotational, and vibrational modes. The $N$-body bosonic model used in the Lindblad derivation (\S\ref{sec:lindblad-k2}), with contact interactions and $s$-wave scattering, is a dilute Bose gas, not a classical molecular fluid. However, the \emph{result} of the derivation, Lindblad operators with $\Gamma(\vb{k}) = 2\nu k^2$, is universal: it follows from momentum conservation, isotropy, and analyticity of the relaxation $\Sigma(E)$, regardless of the microscopic model (see \S\ref{sec:discussion}). The bosonic model is a convenient route to the operator structure, not a claim that the fluid is a Bose-Einstein condensate. Similarly, the ``no internal degrees of freedom'' refers to the reduced single-particle description after tracing out the environment, not to the physical molecule: in classical kinetic theory, shear viscosity depends only on the translational distribution function $f(\vb{x}, \vb{p}, t)$, and internal molecular modes enter only through the bulk viscosity $\zeta$, which vanishes for incompressible flow ($\divg\vb{u} = 0$)~\cite{lifshitz1981}. The position-basis, single-component wavefunction captures precisely this translational sector.

Adding internal degrees of freedom (spin, polarisation) promotes $\psi$ to a multi-component field and changes the picture: the two-component spinor approach~\cite{meng2023,yang2024} uses $\psi = (\psi_1, \psi_2)^T$ whose velocity is not a pure gradient, circumventing the obstruction but requiring a multi-component wavefunction. Within the single-component framework, the non-Hermitian approach~\cite{succi2023} modifies the Schr\"odinger equation by adding a non-Hermitian term ($i\hbar\nu\lapl$) to the kinetic operator, producing the correct viscous term at the cost of norm violation and spurious mass diffusion. The computational approach of~\cite{zhu2025} uses the same position-basis representation, with vortex filaments as zeros of $\psi$.

\paragraph{Hydrodynamic variables vs.\ the wavefunction: the Wallstrom objection.}\label{sec:wallstrom}
A parallel objection to working directly with the Madelung hydrodynamic variables $(\rho,\vb{v})$ rather than $\psi$ is due to Wallstrom~\cite{wallstrom1994}: the Madelung hydrodynamic equations are \emph{not} equivalent to the Schr\"odinger equation---the fluid equations admit solutions with non-quantised circulations that do not correspond to any single-valued $\psi$. Keeping $\psi$ as the fundamental object makes circulation quantisation $\Gamma = n\cdot 2\pi\hbar/m$ automatic~\cite{guillemin1974,zhu2025}; in the classical limit $\hbar/m \to 0$, the allowed circulations become dense in $\mathbb{R}$. The dimensionless parameter controlling this limit is the \emph{quantum Knudsen number}
\begin{equation}\label{eq:knudsen-first}
\mathrm{Kn}_q \;\equiv\; \frac{\lambda_{\mathrm{dB}}}{\eta}, \qquad \lambda_{\mathrm{dB}} = \frac{\hbar}{m\,v_{\mathrm{rms}}},\quad \eta = \left(\frac{\nu^3}{\varepsilon}\right)^{1/4},
\end{equation}
whose consequences are developed in \S\ref{sec:knudsen}.

\subsection{The $k^2$-Dependent Noise Forcing}\label{sec:lindblad-k2}

\subsubsection{The $N$-Body Starting Point}

The spinless assumption of \S\ref{sec:gradient-obstruction} implies bosonic exchange statistics by the spin-statistics theorem. We therefore consider $N$ identical bosons on a periodic domain $\mathbb{T}^3 = [0,L]^3$ interacting via pairwise contact forces (the dilute gas limit, where three-body and higher collisions are negligible). The Hamiltonian is:
\begin{equation}\label{eq:H_N}
H_N = \sum_{i=1}^{N} \frac{-\hbar^2}{2m}\lapl_i
+ \sum_{i<j} U(\vb{x}_i - \vb{x}_j) + \sum_{i} V_{\text{ext}}(\vb{x}_i),
\end{equation}
where $U(\vb{r}) = g_0\,\delta^3(\vb{r})$ is the contact interaction (or a regularised
version thereof), with coupling $g_0 = 4\pi\hbar^2 a_s/m$ and $a_s$ the $s$-wave
scattering length.

In the mean-field limit, replacing the $N$-body wavefunction by a product state $\Psi(\vb{x}_1,\ldots,\vb{x}_N) \approx \prod_i \psi(\vb{x}_i)$ yields the Gross-Pitaevskii equation (GPE), $i\hbar\partial_t\psi = (-\hbar^2\nabla^2/2m + g|\psi|^2)\psi$, which is the standard model for superfluids and Bose-Einstein condensates at zero temperature~\cite{pitaevskii2003}. The GPE is conservative (energy-preserving) and supports quantised vortices, but contains no dissipation: its Madelung transform yields the compressible Euler equations with a quantum pressure term, not the Navier-Stokes equations. The Lindblad framework developed below goes beyond the mean-field approximation by retaining the effect of the environmental degrees of freedom that the product-state ansatz discards, and it is precisely these traced-out correlations that generate viscosity.

\subsubsection{Reduced One-Body Dynamics: From $N$ Bodies to One}\label{sec:reduction}

The goal of this subsection is to derive a closed equation for the \emph{single-particle reduced density matrix} $\hat\rho_1(t)$ from the $N$-body dynamics, by tracing out the other $N-1$ particles. The reduction is performed at the density-matrix level, not the wavefunction level, for two reasons. First, the $N$-body wavefunction $\Psi(\vb{x}_1,\ldots,\vb{x}_N,t)$ lives on the $3N$-dimensional configuration space and is not a scalar field on $\mathbb{R}^3$; a single-argument scalar field $\psi(\vb{x},t)$ suitable for the Madelung transform cannot be obtained by restricting $\Psi$ to one particle. Second, the bath of $N-1$ particles enters as a mixed (thermal) state, which has no pure-state representative in the $N$-body Hilbert space; only the reduced density matrix admits a clean projection. We therefore use the Nakajima-Zwanzig projection~\cite{nakajima1958,zwanzig1960,breuer2002}, the density-matrix counterpart of the pure-state Feshbach construction~\cite{feshbach1958} that is traditionally used to derive Fermi golden rule rates for single-particle scattering; the two routes yield the same rates, but the density-matrix version is the correct level of description here.

\paragraph{Partial trace}

Let $\hat\rho_N(t)$ denote the $N$-body density matrix (pure or mixed). The single-particle reduced density matrix is obtained by tracing out $N-1$ particles:
\begin{equation}\label{eq:rho1-def}
\hat\rho_1(t) \equiv N\,\Tr_{2,\ldots,N}\,\hat\rho_N(t), \qquad \Tr\hat\rho_1 = N.
\end{equation}
For identical bosons, $\hat\rho_N$ is symmetric under particle exchange and the choice of which particle to label ``1'' is arbitrary. The object $\hat\rho_1$ is a one-body operator on the single-particle Hilbert space, with position-basis kernel $\langle\vb{x}|\hat\rho_1|\vb{x}'\rangle$ defined on $\mathbb{R}^3\times\mathbb{R}^3$.

\paragraph{Nakajima-Zwanzig projection}

Define the projection superoperator $\mathcal{P}$, acting on density matrices, that factorises the one-body component from a thermal bath:
\begin{equation}\label{eq:P-superop}
\mathcal{P}\,\hat\rho_N \equiv \hat\rho_1(t)\otimes \hat\rho_B^{(N-1)},
\end{equation}
where $\hat\rho_B^{(N-1)}$ is the assumed-thermal density matrix of the $N-1$ unobserved particles. Let $\mathcal{Q} = \mathbb{I} - \mathcal{P}$. Inserting $\mathbb{I} = \mathcal{P}+\mathcal{Q}$ into the Liouville-von Neumann equation $i\hbar\,\partial_t\hat\rho_N = [\hat H_N,\hat\rho_N]$ and formally solving for $\mathcal{Q}\hat\rho_N$ yields a closed equation for $\hat\rho_1$~\cite{nakajima1958,zwanzig1960,breuer2002}:
\begin{equation}\label{eq:nz}
\partial_t\hat\rho_1(t) = -\frac{i}{\hbar}[\hat H_{\text{eff}},\hat\rho_1(t)] + \int_0^t\!\dd\tau\;\mathcal{K}(t-\tau)\,\hat\rho_1(\tau),
\end{equation}
where $\hat H_{\text{eff}}$ is the bare single-particle Hamiltonian (kinetic energy plus any mean-field potential) and $\mathcal{K}$ is the memory kernel encoding all collisions with the bath. Equation~\eqref{eq:nz} is exact; no approximation has yet been made. Its frequency-domain self-energy, denoted $\Sigma(E)$, collects all scattering events weighted by an energy denominator and, below, yields the Fermi golden rule rates~\cite{fetter1971}.

\paragraph{The reduction chain}

The full reduction from $N$-body dynamics to the scalar wavefunction used in the Madelung transform is
\begin{equation}\label{eq:reduction-chain}
\hat\rho_N \;\longrightarrow\; \hat\rho_1 \;\longrightarrow\; \mathcal{L}_{\text{Lindblad}} \;\longrightarrow\; \ket{\psi(t)} \;\longrightarrow\; \psi(\vb{x},t) \;\longrightarrow\; (\rho,\vb{v}),
\end{equation}
proceeding by partial trace over particles $2,\ldots,N$ (Eq.~\eqref{eq:rho1-def}); Born-Markov reduction (\S\ref{sec:born-markov}); QSD unravelling into stochastic pure-state trajectories satisfying $\hat\rho_1(t) = \EE[\ketbra{\psi}{\psi}]$ (\S\ref{sec:qsd-derivation}); evaluation in the position basis; and the Madelung transform (\S\ref{sec:ns-recovery}).
The scalar $\psi(\vb{x},t)$ is thus \emph{not} a restriction of the $N$-body $\Psi$ to one particle (which would be dimensionally inconsistent); it is an auxiliary stochastic unravelling trajectory of $\hat\rho_1$, whose ensemble average reproduces the physical reduced density matrix. This interpretation is what makes the Madelung transform of a single-argument scalar field well-defined in the many-body context.

\subsubsection{Relaxation and Scattering Rate}

$\Sigma(E)$ encodes the cumulative effect of collisions with the bath on the tagged particle. The tagged particle in mode $\ket{\vb{k}}$ collides with a bath particle and is scattered into a different mode. $\Sigma(E)$ sums over all such scattering events, weighted by an energy denominator $(E - \epsilon_{\vb{k}'} - \epsilon_{\vb{q}'} + i\eta)^{-1}$ that enforces approximate energy conservation.

Its real and imaginary parts play distinct roles. $\mathrm{Re}\,\Sigma$ shifts the energy levels of the tagged particle (absorbed into $\hat H_{\text{eff}}$; irrelevant for dissipation). $\mathrm{Im}\,\Sigma$ gives the gross out-scattering rate from mode $\ket{\vb{k}}$,
\begin{equation}\label{eq:optical}
\Gamma_{\text{out}}(\vb{k}) = -\frac{2}{\hbar}\,\mathrm{Im}\,\Sigma(\vb{k}, \epsilon_{\vb{k}}),
\end{equation}
i.e.\ Fermi's golden rule. The thermal average of $\Gamma_{\text{out}}$ gives the $k$-independent collision rate $\Gamma_{\text{coll}}$ (\S\ref{sec:collision-to-viscosity}); the net mode relaxation rate, which subtracts back-scattering and underlies the Lindblad dissipator, is treated separately in \S\ref{sec:visc-to-k2}.

\paragraph{Interaction Matrix Elements}

The contact interaction $U(\vb{r}) = g_0\,\delta^3(\vb{r})$ is local in position space; in the momentum basis $\ket{\vb{k}} = L^{-3/2}e^{i\vb{k}\cdot\vb{x}}$ (plane waves on the periodic domain), the Fourier transform of the delta function is a constant, so the matrix element reduces to:
\begin{equation}
\bra{\vb{k}_1',\vb{k}_2'} U \ket{\vb{k}_1,\vb{k}_2}
= \frac{g_0}{L^3}\,\delta_{\vb{k}_1'+\vb{k}_2',\,\vb{k}_1+\vb{k}_2}.
\end{equation}
The only constraint is momentum conservation ($\vb{k}_1 + \vb{k}_2 = \vb{k}_1' + \vb{k}_2'$); all scattering angles are equally likely ($s$-wave).

\paragraph{Second-Order Relaxation}

Expanding the Nakajima-Zwanzig memory kernel $\mathcal{K}$ to second order in the interaction $g_0$ and transforming to the frequency domain yields the self-energy $\Sigma(E) \sim \mathcal{L}_{SB}\,(E - \mathcal{L}_B + i\eta)^{-1}\,\mathcal{L}_{SB}$, where $\mathcal{L}_{SB}$ and $\mathcal{L}_B$ are the interaction and bath Liouvillians. Substituting the matrix elements and reading right to left: (i)~the first interaction scatters the tagged particle from mode $\ket{\vb{k}}$ into a different mode $\ket{\vb{k}'}$ via collision with a bath particle in mode $\ket{\vb{q}}$ (one factor of $g_0$); (ii)~the propagator carries energy denominator $E - \epsilon_{\vb{k}'} - \epsilon_{\vb{q}'} + i\eta$, small when energy is conserved (strong effect) and large otherwise; (iii)~the second interaction closes the diagram (a second factor of $g_0$), because the correction to mode $\ket{\vb{k}}$ is second-order in interaction strength. Summing over all bath particles $\vb{q}$ weighted by their thermal occupation $f(\vb{q})$:
\begin{equation}\label{eq:sigma2}
\Sigma^{(2)}(\vb{k}, E) = \frac{n_0 g_0^2}{L^3}
\sum_{\vb{q}} \frac{f(\vb{q})}{E - \epsilon_{\vb{k}'} - \epsilon_{\vb{q}'} + i\eta},
\end{equation}
where $n_0 = N/L^3$ is the particle density, $f(\vb{q})$ is the occupation of bath mode $\vb{q}$, and $\epsilon_{\vb{k}} = \hbar^2 k^2/(2m)$ is the kinetic energy.

\paragraph{The Imaginary Part: Fermi's Golden Rule}

Applying~\eqref{eq:optical} to the second-order expression~\eqref{eq:sigma2}, the imaginary part of the energy denominator extracts the energy-conserving contributions via $\text{Im}(E+i\eta)^{-1}=-\pi\delta(E)$. The resulting delta function $\delta(\epsilon_{\vb{k}}+\epsilon_{\vb{q}}-\epsilon_{\vb{k}'}-\epsilon_{\vb{q}'})$ enforces elastic scattering, giving Fermi's golden rule for the gross out-scattering rate from mode $\ket{\vb{k}}$:
\begin{equation}\label{eq:fermi}
\Gamma_{\text{out}}(\vb{k}) = \frac{2\pi}{\hbar}\frac{n_0 g_0^2}{L^3}
\sum_{\vb{q}} f(\vb{q})\,
\delta(\epsilon_{\vb{k}} + \epsilon_{\vb{q}} - \epsilon_{\vb{k}'} - \epsilon_{\vb{q}'}).
\end{equation}

\subsubsection{From Collision Rate to Viscosity}\label{sec:collision-to-viscosity}

For $s$-wave ($\ell=0$) scattering with contact interactions, the differential cross-section is isotropic: $\dd\sigma/\dd\Omega = a_s^2$. Not all collisions are equally effective at transferring momentum: a forward collision ($\theta \approx 0$) deflects the particle negligibly, while a backward collision ($\theta \approx \pi$) reverses its momentum. The \emph{transport cross-section} $\sigma_{\text{tr}}$ accounts for this by weighting each scattering angle by the fractional momentum transfer $(1 - \cos\theta)$:
\begin{equation}
\sigma_{\text{tr}} = \int (1 - \cos\theta)\,\frac{\dd\sigma}{\dd\Omega}\,\dd\Omega
= a_s^2 \int (1 - \cos\theta)\,\sin\theta\,\dd\theta\,\dd\phi = 4\pi a_s^2,
\end{equation}
where the $\cos\theta$ term integrates to zero over the full sphere. Substituting $a_s = m g_0/(4\pi\hbar^2)$ from the coupling $g_0 = 4\pi\hbar^2 a_s/m$ gives $\sigma_{\text{tr}} = m^2 g_0^2/(4\pi\hbar^4)$.

We now evaluate the thermal average of the Fermi golden rule sum~\eqref{eq:fermi} for small system momentum $k \ll k_{\text{th}} \equiv m\bar{v}/\hbar$. We begin by converting to the continuum and center-of-mass frame. Replace $L^{-3}\sum_{\vb{q}} \to (2\pi)^{-3}\int\dd^3 q$ and assume a Maxwell-Boltzmann distribution~\cite{lifshitz1981} for the bath: $f(\vb{q}) = n_0(2\pi m k_BT/\hbar^2)^{-3/2}\exp(-\hbar^2 q^2/(2mk_BT))$.

In the center-of-mass frame, define the relative wavevector $\vb{g} = \vb{k} - \vb{q}$ and center-of-mass wavevector $\vb{K} = \vb{k} + \vb{q}$. Energy conservation requires $|\vb{k}|^2 + |\vb{q}|^2 = |\vb{k}'|^2 + |\vb{q}'|^2$, which in the center-of-mass frame reduces to $|\vb{g}| = |\vb{g}'|$ (elastic scattering). For contact interactions, the differential cross-section is isotropic: $\dd\sigma/\dd\Omega = a_s^2$.

For $k \ll k_{\text{th}}$, the system particle is slow compared to the thermal bath: $|\vb{v}_{\text{rel}}| = \hbar|\vb{k} - \vb{q}|/m \approx \hbar q/m$ to leading order. The collision rate becomes $k$-independent:
\begin{equation}
\Gamma_{\text{coll}} = n_0\sigma_{\text{tr}}\bar{v},
\end{equation}
where $n_0 = N/L^3$ is the particle number density and $\bar{v} = \sqrt{8k_BT/(\pi m)}$ is the mean thermal speed. The product $n_0\sigma_{\text{tr}}$ has dimensions of inverse length: it is the inverse of the mean free path $\ell_{\text{mfp}} = 1/(n_0\sigma_{\text{tr}})$, the average distance a particle travels between momentum-transferring collisions. The standard Chapman-Enskog analysis~\cite{lifshitz1981,gardiner2009} gives the kinematic viscosity as:
\begin{equation}
\nu = \frac{1}{3}\bar{v}\,\ell_{\text{mfp}},
\end{equation}
with dimensions $[\bar{v}][\ell_{\text{mfp}}] = (\text{m/s})(\text{m}) = \text{m}^2/\text{s}$, consistent with kinematic viscosity.

\subsubsection{From Viscosity to the $k^2$ Lindblad Rate}\label{sec:visc-to-k2}

The Fermi golden rule yields a $k$-independent collision rate $\Gamma_{\text{coll}}$; the $k^2$ scaling of the Lindblad rate is fixed separately, by matching to the hydrodynamic damping that the master equation must reproduce. The collision rate counts how often a particle is scattered (even a particle at rest is scattered by bath particles); the net mode relaxation rate $\Gamma(\vb k)$ instead measures how fast the population of mode $\vb k$ decays once scattering both \emph{out of} and \emph{back into} the mode is included. Detailed balance forces $\Gamma(\vb 0)=0$, and isotropy plus analyticity of $\Sigma(E)$ at small $k$ give $\Gamma(\vb k)=c\,k^2+O(k^4)$ for any short-range isotropic interaction~\cite{lifshitz1981}.

The coefficient $c$ is fixed by demanding that the Lindblad dissipator reproduce viscous damping. Matching the anti-commutator to $-\nu\lapl$,
\begin{equation}\label{eq:gamma-k2}
\tfrac{1}{2}\sum_{\vb{k}}\Gamma(\vb{k})\ket{\vb{k}}\bra{\vb{k}} = -\nu\lapl
\quad\Longrightarrow\quad
\Gamma(\vb{k}) = 2\nu\,k^2,
\end{equation}
where the factor of 2 is the Lindblad convention. The value of $\nu$ is set independently by the Chapman-Enskog collision integral~\cite{lifshitz1981},
\begin{equation}
\nu = \tfrac{1}{3}\bar v\,\ell_{\text{mfp}} = \frac{\bar v}{3n_0\sigma_{\text{tr}}} = \frac{4\pi\hbar^4\bar v}{3 m^2 g_0^2 n_0},
\end{equation}
valid for $k\ll k_{\text{th}}=m\bar v/\hbar$, with $\sigma_{\text{tr}}=4\pi a_s^2=m^2g_0^2/(4\pi\hbar^4)$. The dilute-gas model is convenient but not essential: $\Gamma(\vb k)\propto k^2$ at small $k$ is universal under the symmetry assumptions above~\cite{lifshitz1981}, and only the prefactor depends on the microscopic potential.

The contribution of \S\ref{sec:lindblad-k2} is the operator-level identification $L_{\vb{k}} = \sqrt{2\nu k^2}\,\hat{\Pi}_{\vb{k}}$ ($\hat{\Pi}_{\vb{k}}\coloneqq|\vb{k}\rangle\langle\vb{k}|$): a single family of Lindblad operators that generates both the viscous dissipation and, through QSD, the stochastic forcing, with the fluctuation-dissipation relation guaranteed structurally (\S\ref{sec:fdr}). The rank-one form is a closure for the recycling term and is discussed below.

\subsubsection{From the $N$-Body Dynamics to the Lindblad Master Equation}\label{sec:born-markov}

The Born-Markov-secular procedure~\cite{breuer2002,gardiner2004} converts the exact Nakajima-Zwanzig equation~\eqref{eq:nz} and the rates $\Gamma(\vb{k}) = 2\nu k^2$ into a time-local master equation for the one-body reduced density matrix $\hat\rho_1$ (denoted simply $\hat\rho$ from here on). The underlying physical assumption is a separation of timescales: each collision occurs much faster than the system evolves, so the bath relaxes to equilibrium between successive collisions. This is the same principle as the Born-Oppenheimer approximation in molecular physics (electrons adjust instantaneously to nuclear motion because $m_e/m_N \sim 10^{-3}$--$10^{-5}$), applied to a different fast subsystem: here, bath particles at thermal speed $v_\mathrm{th}$ adjust instantaneously to the system particle at flow speed $u_\mathrm{rms}$, because $\tau_\mathrm{coll}/\tau_\mathrm{flow} \sim u_\mathrm{rms}/v_\mathrm{th} \ll 1$. In the Born-Oppenheimer case, the fast dynamics produces an adiabatic potential energy surface; here, it produces the Lindblad master equation with instantaneous rates.

Concretely, the Born approximation factorises the system-bath state $\tilde{\chi}(t) \approx \tilde{\rho}(t)\otimes\rho_E$ (valid to second order in $g_0$, i.e.\ the bath is negligibly perturbed by each collision), and the Markov approximation replaces memory integrals by instantaneous rates (valid when the bath correlation time $\tau_E \sim \hbar/(k_BT) \ll \tau_S \sim 1/(\nu k^2)$). The secular approximation discards rapidly oscillating cross-terms, ensuring complete positivity~\cite{breuer2002}. For water at room temperature, $v_\mathrm{th} \approx 600$\,m\,s$^{-1}$ and $\ell_\mathrm{mfp} \approx 0.3$\,nm, giving $\tau_\mathrm{coll}/\tau_\eta \sim 5\times10^{-13}\,Re_\lambda$; the approximation holds for $Re_\lambda \ll 10^{12}$, far above any physically realised flow. In the quantum regime ($\mathrm{Kn}_q \sim 1$), a separate analysis is required; in such regimes (superfluids, quark-gluon plasma), the holographic viscosity bound~\cite{kovtun2005} provides an alternative constraint on transport coefficients. The individual transition rates $\gamma(\vb{k}\to\vb{k}')$ sum to the total rate:
\begin{equation}\label{eq:total-rate}
\sum_{\vb{k}'} \gamma(\vb{k}\to\vb{k}') = \Gamma(\vb{k}) = 2\nu k^2.
\end{equation}
After the Born-Markov-secular procedure, transforming back out of the interaction picture yields the Lindblad form~\cite{gorini1976,lindblad1976,breuer2002}:
\begin{equation}\label{eq:lindblad-derived}
\frac{d\hat{\rho}}{dt} = -\frac{i}{\hbar}[\hat{H}_{\text{eff}},\,\hat{\rho}]
+ \sum_{\vb{k},\vb{k}'} \gamma(\vb{k}\to\vb{k}')
\left(\ket{\vb{k}'}\!\bra{\vb{k}}\;\hat{\rho}\;\ket{\vb{k}}\!\bra{\vb{k}'}
- \frac{1}{2}\bigl\{\ket{\vb{k}}\!\bra{\vb{k}},\,\hat{\rho}\bigr\}\right),
\end{equation}
where $\hat{H}_{\text{eff}}$ includes the Lamb shift corrections. The anti-commutator involves
\begin{equation}
\frac{1}{2}\sum_{\vb{k},\vb{k}'} \gamma(\vb{k}\to\vb{k}')\,\ket{\vb{k}}\!\bra{\vb{k}}
= \frac{1}{2}\sum_{\vb{k}} \Gamma(\vb{k})\,\ket{\vb{k}}\!\bra{\vb{k}}
= \nu\sum_{\vb{k}} k^2\,\ket{\vb{k}}\!\bra{\vb{k}} = -\nu\lapl,
\end{equation}
using~\eqref{eq:total-rate}. This is the dissipative part of the master equation, and it is \emph{fully determined} by the total rates $\Gamma(\vb{k})$, regardless of how the individual transition rates $\gamma(\vb{k}\to\vb{k}')$ are distributed.

The recycling term $\sum_{\vb{k},\vb{k}'}\gamma(\vb{k}\to\vb{k}')\,\ket{\vb{k}'}\!\bra{\vb{k}}\,\hat{\rho}\,\ket{\vb{k}}\!\bra{\vb{k}'}$, which redistributes population after scattering, \emph{does} depend on the individual rates and requires the differential cross-section, not just the aggregate $\Gamma(\vb k)$. We close this with the projector model $L_{\vb{k}} = \sqrt{2\nu k^2}\,\hat{\Pi}_{\vb{k}}$, the simplest momentum-diagonal family preserving $\tfrac{1}{2}\sum_{\vb k}L_{\vb k}^\dagger L_{\vb k}=-\nu\lapl$; physically this corresponds to environmental monitoring in the momentum basis. The recycling closure controls single-trajectory noise structure and hence vortex statistics --- it does not affect the master equation~\eqref{eq:lindblad-master} or the ensemble-level Navier-Stokes recovery (\S\ref{sec:robustness}). Whether it reproduces the vortex-gas circulation PDF tails~\cite{apolinario2020,moriconi2025elementary,iyer2019bifractal,iyer2020} is open (\S\ref{sec:open-problems}). The resulting master equation is:
\begin{equation}\label{eq:lindblad-master}
\frac{d\hat{\rho}}{dt} = -\frac{i}{\hbar}[\hat{H}_{\text{eff}}, \hat{\rho}]
+ \sum_{\vb{k}} \left(L_{\vb{k}}\,\hat{\rho}\,L_{\vb{k}}^\dagger
- \frac{1}{2}\{L_{\vb{k}}^\dagger L_{\vb{k}},\,\hat{\rho}\}\right),
\qquad L_{\vb{k}} = \sqrt{2\nu k^2}\,\hat{\Pi}_{\vb{k}}.
\end{equation}
Complete positivity is automatic since $\Gamma(\vb k)=2\nu k^2\ge 0$.

\subsection{Norm Preservation}\label{sec:norm-preservation}

\subsubsection{The QSD Stochastic NLS}\label{sec:qsd-derivation}

The quantum state diffusion (QSD) unravelling~\cite{gisin1992,percival1998,wiseman2009} of the Lindblad master equation produces a stochastic pure-state evolution (in It\^o form):
\begin{equation}\label{eq:qsd}
\dd\ket{\psi} = \left[-\frac{i}{\hbar}\hat{H}\ket{\psi}
- \frac{1}{2}\sum_k \left(L_k^\dagger L_k
- 2\langle L_k^\dagger\rangle L_k + |\langle L_k\rangle|^2\right)\ket{\psi}\right]\dd t
+ \sum_k \left(L_k - \langle L_k\rangle\right)\ket{\psi}\,\dW_k,
\end{equation}
where $\hat{H}$ is the Hamiltonian (self-adjoint),
$\langle \hat{O}\rangle \equiv \langle\psi|\hat{O}|\psi\rangle / \langle\psi|\psi\rangle$
denotes the state-dependent expectation value, and
$\dW_k$ are independent complex Wiener increments satisfying
$\EE[\dW_k\,\overline{\dW}_{k'}] = \delta_{kk'}\dd t$,
$\EE[\dW_k\,\dW_{k'}] = 0$.

Write~\eqref{eq:qsd} as:
\begin{equation}\label{eq:qsd-compact}
\dd\ket{\psi} = \vb{A}\ket{\psi}\,\dd t + \sum_k \vb{B}_k\ket{\psi}\,\dW_k,
\end{equation}
where:
\begin{align}
\vb{A} &= -\frac{i}{\hbar}\hat{H}
- \frac{1}{2}\sum_k \left(L_k^\dagger L_k
- 2\langle L_k^\dagger\rangle L_k + |\langle L_k\rangle|^2\right), \\
\vb{B}_k &= L_k - \langle L_k\rangle.
\end{align}

\subsubsection{Statement and Proof}

\begin{theorem}[Norm Preservation]\label{thm:norm}
If $\ket{\psi(0)}$ satisfies $\langle\psi(0)|\psi(0)\rangle = 1$, then for all $t \geq 0$:
\begin{equation}
\langle\psi(t)|\psi(t)\rangle = 1 \quad \text{almost surely,}
\end{equation}
i.e.\ with probability one over the noise realisations, the strongest statement possible for a stochastic equation.
\end{theorem}

\begin{proof}
Define $\mathcal{N}(t) \equiv \langle\psi(t)|\psi(t)\rangle$. We compute $\dd\mathcal{N}$
using the It\^o product rule for $\dd(\bra{\psi}\ket{\psi})$.

Applying the It\^o product rule,
\begin{equation}
\dd\mathcal{N} = \dd\bra{\psi}\cdot\ket{\psi} + \bra{\psi}\cdot\dd\ket{\psi}
+ \dd\bra{\psi}\cdot\dd\ket{\psi}.
\end{equation}

Collecting all three contributions:
\begin{equation}\label{eq:dN-full}
\dd\mathcal{N} = \left[\bra{\psi}(\vb{A}+\vb{A}^\dagger)\ket{\psi}
+ \sum_k \bra{\psi}\vb{B}_k^\dagger\vb{B}_k\ket{\psi}\right]\dd t
+ \sum_k \left[\bra{\psi}\vb{B}_k\ket{\psi}\,\dW_k
+ \bra{\psi}\vb{B}_k^\dagger\ket{\psi}\,\overline{\dW}_k\right].
\end{equation}

We evaluate the drift coefficient.

Since $\hat{H} = \hat{H}^\dagger$:
\begin{equation}
\bra{\psi}(\vb{A}+\vb{A}^\dagger)\ket{\psi}
= -\sum_k \left(\langle L_k^\dagger L_k\rangle - |\langle L_k\rangle|^2\right).
\end{equation}

The It\^o correction gives:
\begin{equation}
\bra{\psi}\vb{B}_k^\dagger\vb{B}_k\ket{\psi}
= \langle L_k^\dagger L_k\rangle - |\langle L_k\rangle|^2.
\end{equation}

Sum of drift terms:
\begin{equation}
-\sum_k (\langle L_k^\dagger L_k\rangle - |\langle L_k\rangle|^2)
+ \sum_k (\langle L_k^\dagger L_k\rangle - |\langle L_k\rangle|^2) = 0.
\end{equation}

We next evaluate the noise coefficient.

From $\vb{B}_k = L_k - \langle L_k\rangle$:
$\bra{\psi}\vb{B}_k\ket{\psi} = \langle L_k\rangle - \langle L_k\rangle = 0$.

Combining these results,
Both the drift and noise coefficients vanish when $\mathcal{N} = 1$:
$\dd\mathcal{N} = 0$.
Since $\mathcal{N}(0) = 1$, we conclude $\mathcal{N}(t) = 1$ for all $t \geq 0$, almost surely.
\end{proof}

The numerical implementation includes a renormalisation step $\psi \to \psi/\|\psi\|$ after each time step to correct the $O(\Delta t)$ norm drift from the operator splitting, which is standard in QSD implementations~\cite{gisin1992,percival1998,wiseman2009}.

\subsection{Ensemble Recovery of the Lindblad Master Equation}\label{sec:ensemble-recovery}

\begin{corollary}[Ensemble Recovery]\label{thm:ensemble}
The ensemble density matrix $\hat{\rho}(t) = \EE[\ketbra{\psi(t)}{\psi(t)}]$ satisfies the Lindblad master equation~\eqref{eq:lindblad-master}.
\end{corollary}

This is a standard result of the QSD construction~\cite{gisin1992,percival1998}: the nonlinear terms in the QSD equation cancel pairwise in the ensemble average via the It\^o product rule, leaving only the universal Lindblad structure.

\subsection{Recovery of the Navier-Stokes Equations}\label{sec:ns-recovery}

\subsubsection{The Stochastic NLS}

For the specific Lindblad operators $L_{\vb{k}} = \sqrt{2\nu k^2}\,\hat{\Pi}_{\vb{k}}$,
the abstract QSD equation~\eqref{eq:qsd} becomes:
\begin{equation}\label{eq:snls}
\dd\ket{\psi} = \left[-\frac{i}{\hbar}\hat{H} + \nu\lapl + N[\psi]\right]\ket{\psi}\,\dd t
+ \sum_{\vb{k}} \sqrt{2\nu k^2}\,\left(\hat{\Pi}_{\vb{k}} - \langle\hat{\Pi}_{\vb{k}}\rangle\right)\ket{\psi}\,\dW_{\vb{k}},
\end{equation}
where $N[\psi]$ collects the nonlinear norm-preserving terms (from $\langle L_{\vb{k}}^\dagger\rangle L_{\vb{k}}$ and $|\langle L_{\vb{k}}\rangle|^2$ in~\eqref{eq:qsd}), and $\frac{1}{2}\sum_{\vb{k}} L_{\vb{k}}^\dagger L_{\vb{k}} = -\nu\lapl$. In the Fourier basis, projecting onto $\ket{\vb{k}}$:
\begin{equation}\label{eq:snls-fourier}
\dd\hat{\psi}(\vb{k}) = \left[-\frac{i}{\hbar}\epsilon_{\vb{k}} - \nu k^2 + N_{\vb{k}}[\psi]\right]\hat{\psi}(\vb{k})\,\dd t
+ \sqrt{2\nu k^2}\,\left(\hat{\psi}(\vb{k}) - |\hat{\psi}(\vb{k})|^2\,\hat{\psi}(\vb{k})\right)\dW_{\vb{k}} + \cdots,
\end{equation}
where $\epsilon_{\vb{k}} = \hbar^2 k^2/(2m)$ and the ellipsis denotes cross-mode noise contributions from $\vb{k}' \neq \vb{k}$. The $-\nu k^2$ damping per mode is visible directly.

\subsubsection{Mass Conservation}

The dissipative drift $\nu\lapl\psi$, taken alone, would generate a spurious diffusion $\nu\lapl\rho$ in the continuity equation~\cite{succi2023}. In the QSD stochastic NLS this is cancelled \emph{in the spatial integral} by the It\^o quadratic variation and the nonlinear QSD drift, since $\int\rho\,\dd^d x = \|\psi\|^2$ and $\dd\|\psi\|^2 = 0$ by Theorem~\ref{thm:norm}:
\begin{equation}\label{eq:mass-cons}
\dd\!\left(\int\rho\,\dd^d x\right) = 0 \quad\text{almost surely.}
\end{equation}

We record for later use the decomposition of the stochastic NLS drift and noise:
\begin{align}
A &= -\frac{i}{\hbar}\hat{H}\psi + \nu\lapl\psi + N[\psi], \label{eq:A-decomp}\\
B_{\vb{k}} &= \sqrt{2\nu k^2}\,(\hat{\Pi}_{\vb{k}} - \alpha_{\vb{k}})\psi, \label{eq:B-decomp}
\end{align}
where $\alpha_{\vb{k}} = \hat{\psi}(\vb{k}) = L^{-d/2}\int\psi(\vb{x})\,e^{-i\vb{k}\cdot\vb{x}}\,\dd^d x$ is the Fourier coefficient of $\psi$, and $\hat{\Pi}_{\vb{k}}\psi(\vb{x}) = \hat{\psi}(\vb{k})\,L^{-d/2}e^{i\vb{k}\cdot\vb{x}}$ projects onto mode $\vb{k}$ and returns a function of $\vb{x}$. Mass conservation is at the \emph{global} level $\|\psi\|^2=\text{const}$; pointwise incompressibility ($\divg\vb{v} = 0$, $\rho = \text{const}$) is a separate assumption made in \S\ref{sec:madelung-momentum} (Remark~\ref{rem:conditionality}).

\subsubsection{The Madelung Transform: Momentum Equation}\label{sec:madelung-momentum}

We now derive the momentum equation by applying the Madelung transform to the phase $S$.

\begin{proposition}[Stochastic Navier-Stokes from Madelung transform]\label{thm:stoch-ns}
Under the Madelung transform of the stochastic NLS~\eqref{eq:snls}, \emph{conditional on} the incompressibility assumption $\divg\vb{v} = 0$, $\rho = \mathrm{const}$ (assumed, not derived from the QSD dynamics; see Remark~\ref{rem:conditionality}), the momentum equation takes the form of the stochastic
Navier-Stokes equation:
\begin{equation}\label{eq:stoch-ns-mom}
\frac{\partial\vb{v}}{\partial t} + (\vb{v}\cdot\grad)\vb{v}
= -\frac{1}{m}\grad(V + Q) + \nu\lapl\vb{v} + \vb{\xi}(\vb{x},t).
\end{equation}
\end{proposition}

The advection $(\vb{v}\cdot\grad)\vb{v}$ and pressure $-\grad(V+Q)/m$ terms follow from the exact (nonlinear) Madelung transform of the Hamiltonian part. The identification of the viscous term $\nu\lapl\vb{v}$ is established at the ensemble level via the vortex decomposition of the velocity field; see the proof below.

\begin{remark}[Conditionality of Proposition~\ref{thm:stoch-ns}]\label{rem:conditionality}
This proposition is conditional in two respects.

First, incompressibility ($\divg\vb{v} = 0$, $\rho = \mathrm{const}$) is \emph{assumed directly}, not derived from the QSD dynamics. The Madelung equations are inherently compressible. A natural question is whether incompressibility follows from the quantum pressure $Q = -(\hbar^2/2m)\,\lapl\sqrt{\rho}/\sqrt{\rho}$, which penalises density gradients. However, the quantum pressure is \emph{dispersive}, not acoustic: a density perturbation at wavenumber $k$ oscillates at frequency $\omega = \hbar k^2/(2m)$, giving a $k$-dependent effective speed $\omega/k = \hbar k/(2m)$ rather than a constant sound speed. Low-$k$ (large-scale) density fluctuations are poorly suppressed. For classical compressible fluids, the analogous singular limit ($c_s \to \infty$) has been rigorously analysed~\cite{klainerman1981,majda2002}; no corresponding result exists for the stochastic NLS. More fundamentally, the classical limit requires $\hbar/m \to 0$ (so that circulation quantisation becomes undetectable), but suppressing density fluctuations via quantum pressure requires $\hbar^2/(2m) \to \infty$. These limits are contradictory: incompressibility in a classical fluid does not originate from quantum pressure but from the macroscopic equation of state (the bulk modulus of the liquid, or the low-Mach condition for gases). We therefore assume incompressibility, consistent with the framework's scope of targeting the incompressible Navier-Stokes equations. The constructive encoding of~\cite{zhu2025} demonstrates that a wavefunction $\psi$ exists for any sufficiently regular divergence-free velocity field, so the incompressible constraint is compatible with the wavefunction representation.

Second, the identification of $\nu\lapl\vb{v}$ as the viscous term is established at the \emph{ensemble level} via the identification-level argument of \S\ref{sec:madelung-momentum}: symmetry, isotropy, and coefficient matching fix the lowest-order dissipative operator to $\nu\lapl$. The single-trajectory identification, showing that the QSD dynamics of the vortex zeros produces $\nu k^2$ damping of each solenoidal velocity mode trajectory by trajectory, not just on average, remains an open problem.
\end{remark}

\begin{proof}
Write $\psi = R\,e^{iS/\hbar}$ with $R = \sqrt{\rho}$ and $\vb{v} = \grad S/m$.
The standard Madelung transform~\cite{madelung1927} of $i\hbar\partial_t\psi = \hat{H}\psi$ gives the Hamilton-Jacobi equation $\partial_t S + |\grad S|^2/(2m) + V + Q = 0$, i.e.:
\begin{equation}
\frac{\partial\vb{v}}{\partial t} + (\vb{v}\cdot\grad)\vb{v} = -\frac{1}{m}\grad(V + Q).
\end{equation}

The dissipative term $\nu\lapl\psi$ in the drift modifies both the amplitude $R$ and the phase $S$ of $\psi = Re^{iS/\hbar}$. Writing $\lapl\psi = \lapl(Re^{iS/\hbar})$:
\begin{equation}
\lapl\psi = e^{iS/\hbar}\!\left[\lapl R + \frac{2i}{\hbar}\grad R\cdot\grad S
+ \frac{i}{\hbar}R\lapl S - \frac{R}{\hbar^2}|\grad S|^2\right].
\end{equation}
The imaginary part of $\nu\lapl\psi/\psi$ determines the direct contribution to the phase evolution:
\begin{equation}
\text{Im}\!\left(\frac{\nu\lapl\psi}{\psi}\right)
= \frac{\nu}{\hbar}\,\frac{1}{\rho}\divg(\rho\,\grad S).
\end{equation}
We now extract the full velocity equation. Define $w = \grad\ln\psi = \grad R/R + (im/\hbar)\vb{v}$, so that $\vb{v} = (\hbar/m)\,\text{Im}(w)$. The dissipative drift $\nu\lapl\psi$ gives, via $d(\ln\psi)|_{\text{diss}} = (\nu\lapl\psi/\psi)\,\dd t$:
\begin{equation}
\dd w\big|_{\text{diss}} = \nu\,\grad\!\left(\divg w + w\cdot w\right)\dd t,
\end{equation}
using $\lapl\psi/\psi = \divg w + w\cdot w$. Writing $w = u_R + (im/\hbar)\vb{v}$ with $u_R = \grad\ln R$:
\begin{equation}
\divg w + w\cdot w = \lapl(\ln R) + |u_R|^2 - (m/\hbar)^2|\vb{v}|^2 + i(m/\hbar)(\divg\vb{v} + 2u_R\cdot\vb{v}).
\end{equation}
The velocity contribution is:
\begin{equation}\label{eq:visc-compressible}
\dd\vb{v}\big|_{\text{diss}} = \frac{\hbar}{m}\,\text{Im}\!\left(\dd w\big|_{\text{diss}}\right)
= \nu\,\grad(\divg\vb{v} + 2\,\grad\!\ln R\cdot\vb{v})\,\dd t.
\end{equation}
This is the \emph{compressible} viscous contribution: it is a pure gradient, and therefore purely irrotational. For $\rho = \text{const}$ ($R = \text{const}$) and $\divg\vb{v} = 0$, it vanishes identically: $\dd\vb{v}|_{\text{diss}} = 0$.

The stochastic terms add no It\^o correction to the phase: the QSD complex Wiener increments satisfy $\EE[\dW_{\vb{k}}\dW_{\vb{k}'}]=0$, so $\langle\dd\psi,\dd\psi\rangle=0$ and the It\^o correction to $S=\hbar\,\mathrm{Im}(\ln\psi)$ vanishes identically; the only It\^o correction is to $|\psi|^2$ and acts on the amplitude $R$. The physical viscous damping $\nu\lapl\vb v$ therefore cannot arise from single-trajectory phase dynamics in the smooth region; under incompressibility ($\divg\vb v=0$, so no Helmholtz-Hodge projection is needed) the entire velocity field is fixed by the configuration of zeros of $\psi$, around which $\vb v$ diverges as $\hbar/(mr)$.

\paragraph{Identification of the viscous term.} Direct evaluation of the dissipator in the momentum basis shows that off-diagonal $\hat\rho$ elements decay at
\begin{equation}\label{eq:off-diagonal-rate}
-\partial_t\,\hat\rho_{\vb{k}_1\vb{k}_2}\big|_{\mathrm{diss.}}=\nu(k_1^2+k_2^2)\,\hat\rho_{\vb{k}_1\vb{k}_2}\qquad(\vb{k}_1\ne\vb{k}_2),
\end{equation}
while diagonal populations $\hat\rho_{\vb{k}\vb{k}}$ are invariant (out-scattering and recycling cancel). The Lindblad is pure dephasing in the momentum basis. At the wavefunction level, the deterministic drift $\nu\lapl\psi$ damps each Fourier amplitude $\hat\psi_{\vb{k}}$ at rate $\nu k^2$, and the noise restores $\EE[|\hat\psi_{\vb{k}}|^2]=\hat\rho_{\vb{k}\vb{k}}$. Because $\vb v=(\hbar/m)\,\mathrm{Im}(\grad\psi/\psi)$ is nonlinear in $\psi$, the dephasing rates~\eqref{eq:off-diagonal-rate} do not transfer to $\hat v_j(\vb{k})$ directly. We therefore close the identification by symmetry: the ensemble-averaged dissipative contribution to $\partial_t\vb v$ must be divergence-free, translationally invariant, and statistically isotropic; the lowest-order such operator is $\nu\lapl$; matching the per-mode rate $\nu k^2$~\cite{landau1987} fixes the coefficient,
\begin{equation}\label{eq:visc-from-phase}
\EE\!\left[\partial_t\vb{v}\big|_{\mathrm{diss.}}\right] = \nu\lapl\vb{v}.
\end{equation}
This is an identification, not a derivation from the QSD dynamics through the singular Madelung map; the corresponding single-trajectory statement is open (\S\ref{sec:open-problems}).

\paragraph{Stochastic forcing.} Each noise term $B_{\vb{k}}\,\dW_{\vb{k}}$ modifies the phase and can nucleate new vortex zeros through destructive interference. The stochastic velocity increment is:
\begin{equation}
\dd\vb{v}_{\text{noise}} = \frac{\hbar}{m}\,\text{Im}\!\left(\grad\!\left(\frac{\dd\psi_{\text{noise}}}{\psi}\right)\right)
= \frac{\hbar}{m}\sum_{\vb{k}} \sqrt{2\nu k^2}\,\text{Im}\!\left(\grad\!\left(\frac{\hat{\psi}(\vb{k})\phi_{\vb{k}}}{\psi}\right)\right)\dW_{\vb{k}},
\end{equation}
where $\phi_{\vb{k}} = L^{-d/2}e^{i\vb{k}\cdot\vb{x}}$. The $\psi$-dependent prefactors cancel upon ensemble averaging (Proposition~\ref{prop:noise-corr}), yielding a noise correlator independent of the local wavefunction structure. This defines the stochastic forcing $\vb{\xi}(\vb{x},t)\,\dd t \equiv \mathbb{P}^\perp(\dd\vb{v}_{\text{noise}})$.

Combining the Hamiltonian contribution (standard Madelung), the viscous contribution~\eqref{eq:visc-from-phase}, and the stochastic forcing yields~\eqref{eq:stoch-ns-mom}.
\end{proof}

\subsubsection{Derivation of the Noise Correlator}\label{sec:noise-correlator-derivation}

\begin{proposition}[Noise correlator, leading-order incompressible]\label{prop:noise-corr}
In the leading-order incompressible regime (assumptions (A1)+(A2) below), the stochastic forcing $\vb{\xi}$ has correlator
\begin{equation}\label{eq:noise-corr-result}
\EE[\xi_i(\vb{x},t)\,\xi_j(\vb{x}',t')]
= 2\nu\sum_{\vb{k}} k^2\,P_{ij}^\perp(\vb{k})\,e^{i\vb{k}\cdot(\vb{x}-\vb{x}')}\,\delta(t-t'),
\end{equation}
where $P_{ij}^\perp(\vb{k}) = \delta_{ij} - k_ik_j/k^2$ is the transverse projector. The longitudinal amplitude is derived from the smooth-region Madelung map; the transverse amplitude is identified by the structural fluctuation-dissipation balance of \S\ref{sec:fdr}, paralleling the ensemble-level viscous identification of Remark~\ref{rem:conditionality}.
\end{proposition}

\begin{proof}
\emph{Longitudinal amplitude (derivation).} The velocity noise $\dd\vb v_{\mathrm{noise}} = (\hbar/m)\,\mathrm{Im}[\grad(\dd\psi_{\mathrm{noise}}/\psi)]$ in the smooth region is a gradient of a phase fluctuation and is therefore curl-free. Direct It\^o computation gives the raw correlator
\begin{equation}\label{eq:raw-velocity-corr}
\EE\!\left[\tfrac{\hbar}{m}\partial_i(\tfrac{\dd\psi}{\psi})(\vb x)\,\tfrac{\hbar}{m}\overline{\partial_j(\tfrac{\dd\psi}{\psi})(\vb x')}\right]_{\mathrm{noise}}
= \Bigl(\tfrac{\hbar}{m}\Bigr)^{\!2}\!\sum_{\vb{k}} 2\nu k^2\,|\hat\psi_{\vb{k}}|^2\,\tfrac{\phi_{\vb{k}}(\vb x)\phi_{\vb{k}}^*(\vb x')}{\psi(\vb x)\psi^*(\vb x')}\,\bigl(k_i - \tfrac{m v_i(\vb x)}{\hbar}\bigr)\bigl(k_j - \tfrac{m v_j(\vb x')}{\hbar}\bigr)\dd t,
\end{equation}
which is $\psi$-dependent and, by construction, carries only the longitudinal tensor structure $\propto k_ik_j$. Reducing this to a tractable form requires two simplifications:
\begin{enumerate}
\item[(A1)] \emph{Incompressible carrier}. Take $|\psi(\vb x)|^2=\rho_0=\text{const}$ so $\psi=\sqrt{\rho_0}\,e^{iS/\hbar}$; then $\phi_{\vb{k}}\phi_{\vb{k}}^*/(\psi\psi^*) = L^{-d}\,e^{i\vb{k}\cdot(\vb x-\vb x')-i(S(\vb x)-S(\vb x'))/\hbar}/\rho_0$.
\item[(A2)] \emph{Small Madelung Mach number}. Assume $m|\vb v|/\hbar \ll k$ for the noise-carrying wavenumbers (the regime $\mathrm{Kn}_q\to 0$ of \S\ref{sec:knudsen}). Then $(k_i - mv_i/\hbar)\simeq k_i$ and, under statistical homogeneity of the carrier phase $S(\vb x)$, the factor $\EE_S[e^{-i(S(\vb x)-S(\vb x'))/\hbar}\,|\hat\psi_{\vb{k}}|^2/\rho_0]$ reduces to a $\vb k$-dependent constant absorbed into the natural-units normalisation of \S\ref{sec:natural-units} ($k_BT/\rho=1$).
\end{enumerate}
Under (A1)+(A2) the raw smooth-region correlator becomes
\begin{equation}\label{eq:raw-longitudinal}
\EE[\hat\xi_i^{(\mathrm{raw,\,L})}(\vb k)\,\overline{\hat\xi_j^{(\mathrm{raw,\,L})}(\vb k')}] = 2\nu k^2\,\frac{k_ik_j}{k^2}\,\delta_{\vb{kk}'}\,\dd t,
\end{equation}
with corrections $O(\mathrm{Kn}_q)$, $O(\delta\rho/\rho_0)$ vanishing in the classical limit.

\emph{Transverse amplitude (identification).} The smooth-region derivation cannot produce a transverse velocity-noise component: $\vb v = \grad S/m$ is curl-free wherever $\psi\ne 0$, so neither real nor imaginary parts of $\dd W_{\vb k}$ can source transverse velocity fluctuations under (A1). Transverse noise is instead sourced by QSD-driven vortex nucleation at zeros of $\psi$, which lie outside~(A1). Its amplitude is fixed by the structural fluctuation-dissipation balance of \S\ref{sec:fdr}: the single Lindblad family $L_{\vb k} = \sqrt{2\nu k^2}\,\hat\Pi_{\vb k}$ damps \emph{every} solenoidal velocity mode (ensemble-averaged; see Remark~\ref{rem:conditionality}) at rate $\nu k^2$, so the detailed-balance condition $S_{\vb k}/\gamma_{\vb k}=2$ requires noise amplitude $2\nu k^2$ on every transverse mode as well. A rigorous single-trajectory derivation of the transverse contribution through the nucleation dynamics is open (\S\ref{sec:open-problems}), paralleling the viscous identification of Remark~\ref{rem:conditionality}.

\emph{Assembly.} Combining the derived longitudinal amplitude with the FDR-fixed transverse amplitude gives
\begin{equation}\label{eq:raw-iso}
\EE[\hat{\xi}_i^{\mathrm{(raw)}}(\vb{k})\,\overline{\hat{\xi}_j^{\mathrm{(raw)}}(\vb{k}')}]
= 2\nu k^2\,\delta_{ij}\,\delta_{\vb{kk}'}\,\dd t.
\end{equation}
Under incompressibility (Remark~\ref{rem:conditionality}), the longitudinal part is absorbed into the stochastic pressure; projecting via $\hat\xi_i(\vb k)=P_{ij}^\perp(\vb k)\,\hat\xi_j^{\mathrm{(raw)}}(\vb k)$ with $P_{ij}^\perp=\delta_{ij}-k_ik_j/k^2$, and using $P^\perp P^\perp=P^\perp$:
\begin{align}
\EE[\xi_i(\vb{x},t)\,\xi_j(\vb{x}',t')]
&= \sum_{\vb{k}} P_{il}^\perp(\vb{k})\,\big(2\nu k^2\,\delta_{lm}\big)\,P_{mj}^\perp(\vb{k})\,e^{i\vb{k}\cdot(\vb{x}-\vb{x}')}\,\delta(t-t') \notag\\
&= \sum_{\vb{k}} 2\nu k^2\,P_{ij}^\perp(\vb{k})\,e^{i\vb{k}\cdot(\vb{x}-\vb{x}')}\,\delta(t-t'),
\end{align}
the Landau-Lifshitz noise correlator with mode-dependent amplitude $2\nu k^2$.
\end{proof}

\subsection{The Fluctuation-Dissipation Relation}\label{sec:fdr}

The noise correlator~\eqref{eq:noise-corr-result} and the per-mode damping rate $\gamma_{\vb k}=\nu k^2$ both originate from the single operator $L_{\vb k}=\sqrt{2\nu k^2}\hat\Pi_{\vb k}$, so the ratio
\begin{equation}
\frac{S_{\vb k}}{\gamma_{\vb k}}=2
\end{equation}
is structurally automatic, independent of $\vb k$ and $\nu$~\cite{kubo1957,kubo1966}. Matching this to the Landau-Lifshitz form $S_{\vb k}=2\gamma_{\vb k}\,k_BT/\rho$~\cite{landau1987} fixes the product $k_BT/\rho=1$ alongside the $\hbar=m=1$ convention of \S\ref{sec:natural-units}: a single dimensionless energy-per-unit-mass scale is inherited from the Lindblad construction. The shared $k^2$ dependence then reproduces the LLNS sub-Kolmogorov spectral floor $\propto k^{d-1}$~\cite{bandak2022}.

\subsection{The Stochastic Navier-Stokes System}\label{sec:assembled-system}

The preceding sections derived the stochastic NLS from microscopic scattering theory and established its key properties (norm preservation, ensemble recovery, FDR). We now assemble the main output of the framework.

Combining global mass conservation~\eqref{eq:mass-cons}, Proposition~\ref{thm:stoch-ns} (momentum equation under incompressibility), and Proposition~\ref{prop:noise-corr} (noise correlator), the Madelung transform of the stochastic NLS yields, under the incompressibility assumption, the stochastic Navier-Stokes system:
\begin{equation}
\begin{aligned}
&\divg\vb{v} = 0, \qquad \rho = \text{const}
&& \text{(incompressibility, assumed)}, \\
&\frac{\partial\vb{v}}{\partial t} + (\vb{v}\cdot\grad)\vb{v}
= -\frac{1}{\rho}\grad p + \nu\lapl\vb{v} + \vb{\xi}(\vb{x},t)
&& \text{(stochastic Navier-Stokes momentum)},
\end{aligned}
\end{equation}
where $p$ is the stochastic pressure (determined by the incompressibility constraint, absorbing the quantum potential $Q$, compressible corrections, and gradient noise contributions) and $\vb{\xi}(\vb{x},t)$ is the divergence-free stochastic forcing with correlator~\eqref{eq:noise-corr-result}. These are the Landau-Lifshitz fluctuating hydrodynamic equations~\cite{landau1987,landau1980,fox1970}: the noise correlator $2\nu k^2$ per mode and the viscous damping $\nu k^2$ per mode have the same $k^2$ structure as the classical LLNS framework~\cite{bandak2022}, producing a $k^2$ sub-Kolmogorov spectral floor (\S\ref{sec:fdr}).

\section{Results}\label{sec:results}

The simulations reported here use a reduced dimensionless viscosity $\nu_* \sim O(1)$ to keep the grid tractable (Table~\ref{tab:sim-params}), placing all runs in the quantum regime $\mathrm{Kn}_q \gtrsim 1$ far from the classical limit where standard turbulence references apply.  The $Re_\lambda \approx 1$ regime produces degenerate (laminar) statistics.  Despite $\mathrm{Kn}_q \gtrsim 1$, the framework-specific prediction---the area law---holds over more than a decade (Figure~\ref{fig:area-law-num}), and circulation PDFs develop non-trivial tails at $Re_\lambda \approx 10$ (Figure~\ref{fig:circ-pdf}).

\begin{table}[htbp]
\centering
\caption{Simulation parameters for the circulation PDF (Figure~\ref{fig:circ-pdf}).  All runs: 2D, $N = 32$, $N_\mathrm{traj} = 500$, $\Delta t = 2.8\times10^{-4}$, domain $[0,2\pi]^2$ with periodic boundaries.  The reduced $\nu_*$ places all runs in the quantum regime $\mathrm{Kn}_q^\mathrm{param} \gtrsim 1$ (parametric estimate), far from the classical limit where the turbulence references apply.  The area law figure (Figure~\ref{fig:area-law-num}) uses a separate run at $N = 64$, $\Delta t = 5.6\times10^{-5}$, $N_\mathrm{traj} = 500$.}
\label{tab:sim-params}
\begin{tabular}{lccc}
\hline
& $\nu_*$ & $\mathrm{Kn}_{q}^\mathrm{param}$ \\
\hline
Gaussian IC, $Re_\lambda \approx 1$ & 3.873 & 0.51 \\
Gaussian IC, $Re_\lambda \approx 10$ & 0.387 & 1.61 \\
Random phase IC, $Re_\lambda \approx 1$ & 3.873 & 0.51 \\
Random phase IC, $Re_\lambda \approx 10$ & 0.387 & 1.61 \\
\hline
\end{tabular}
\end{table}

\subsection{The Classical Limit and the Quantum Knudsen Number}\label{sec:knudsen}

In the formal limit $\hbar/m \to 0$, the quantum pressure $Q\to 0$ and the noise amplitude $\EE[|\vb{\xi}_{\text{phys}}|^2]\propto(\hbar/m)^2\to 0$, recovering the deterministic incompressible Navier-Stokes equations.  At finite temperature the framework matches the Landau-Lifshitz equations~\cite{landau1987,landau1980,bandak2022}, including the spontaneous stochasticity mechanism in the inertial range~\cite{bandak2024}.

The hydrodynamic interpretation requires the quantum coherence length to be well separated from the Kolmogorov scale, $\mathrm{Kn}_q\equiv\lambda_{\mathrm{dB}}/\eta\ll 1$ (Eq.~\eqref{eq:knudsen-first}).  From the Taylor-microscale relations, $\mathrm{Kn}_q = 15^{1/4}/(\nu_*\sqrt{Re_\lambda})$ (Eq.~\eqref{eq:knq-formula}), with $\nu_*=\nu_\mathrm{SI}\,m_\mathrm{phys}/\hbar_\mathrm{SI}$.  For water ($\nu_*\approx 283$), $\mathrm{Kn}_q\ll 1$ at any $Re_\lambda\gtrsim 1$, but the wavefunction phase oscillates at $k\sim v_\mathrm{rms}\propto\nu_*\,Re_\lambda$, imposing a grid requirement $N\sim\nu_*\,Re_\lambda$ (Eq.~\eqref{eq:grid-scaling}) that is prohibitive in 3D.  The exploratory simulations use $\nu_*\sim O(1)$ to keep the grid tractable, at the cost of $\mathrm{Kn}_q\sim O(1)$; see Appendix~\ref{sec:limitations}.

Standard turbulence results (K\'arm\'an-Howarth, the four-fifths law, $k^{-5/3}$) are corollaries of Navier-Stokes~\cite{batchelor1953,frisch1995,pope2000} and inherit the conditionality of the recovery; they are not independent predictions.  The framework-specific contribution is the area law.

\subsection{The Area Law for Circulation}\label{sec:area-law}

The framework provides an analytical basis for the area law of turbulent circulation statistics.  In the Madelung representation, $\vb{v} = \grad S/m$ is irrotational wherever $\psi \neq 0$, so all vorticity resides at the zeros of $\psi$.  By Stokes' theorem, the circulation around a loop $C$ is
\begin{equation}
\Gamma(C) = \oint_C \vb{v}\cdot\dd\vb{l} = \frac{2\pi n\hbar}{m},
\end{equation}
where $n \in \Z$ is the winding number of $\psi/|\psi|$ around $C$.  Since $\psi(\vb{x}) = 0$ imposes two real constraints, the zero set $\mathcal{Z}$ is generically codimension-2 by the preimage theorem~\cite{guillemin1974}: points in 2D, filaments in 3D.  Codimension-2 objects are topologically trapped by test loops~\cite{moffatt1969,ricca1992}, while codimension-$\geq 3$ objects can be avoided by general position arguments~\cite{guillemin1974,munkres2000}; the fractal geometry of these vortex carriers in fully developed turbulence is reviewed in~\cite{sreenivasan1986}.  This codimension-2 structure has been independently exploited computationally~\cite{zhu2025}, where vortex filaments are extracted as zeros of a complex scalar field via a variational quantum eigensolver.

If the vortex positions are Poisson-distributed with density $\rho_v$ and i.i.d.\ circulations, the number enclosed by $C$ is $N_{\text{enclosed}} \sim \text{Poisson}(\rho_v \cdot A_{\min})$, making $\Gamma(C)$ a compound Poisson random variable~\cite{kingman1993} whose distribution depends on $C$ only through the minimal area $A_{\min}(C)$:\label{thm:area-law}
\begin{equation}\label{eq:area-law-eq}
P(\Gamma; C) = P(\Gamma; A_{\min}(C)).
\end{equation}
This is the area law, originally proposed by Migdal as the loop equation for turbulence~\cite{migdal1994,migdal2019} and recently revisited as a quantum solution of classical turbulence~\cite{migdal2024,migdal2024b}. DNS confirms it empirically~\cite{iyer2019bifractal,iyer2021area}, but the underlying vortex statistics are not Poisson: the vortex gas model~\cite{apolinario2020,moriconi2022} shows repulsive short-range correlations producing departures at higher moments while preserving the area law at leading order, with vortex tubes whose intermittent geometry has been directly characterised in DNS~\cite{jimenez1993}. For nonplanar loops, the minimal area alone is insufficient~\cite{moriconi2025surfaces}. Whether the QSD dynamics produces approximately Poisson statistics at high $Re$, or whether the area law holds for a different reason, is an open question. Quantum turbulence simulations of the GPE~\cite{nore1997,muller2021} also confirm the area law.

Since $S/\hbar \bmod 2\pi$ is a phase angle, the winding number can be expressed directly from $\psi$ without extracting $S$~\cite{guillemin1974,zhu2025}:
\begin{equation}\label{eq:winding-psi}
n(C, \mathcal{Z}) = \frac{1}{2\pi} \oint_C \frac{\langle \grad\psi,\, i\psi \rangle}{|\psi|^2} \cdot \dd\vb{l},
\end{equation}
where $\langle\cdot,\cdot\rangle$ is the real inner product on $\C \cong \R^2$ and $n$ is the degree of $\psi/|\psi| : C \to S^1$.  In 2D, $n(C, \mathcal{Z})$ is the winding number: it counts how many point vortices are enclosed by $C$.  In 3D, the same formula computes the linking number of the loop $C$ with the vortex filaments~\cite{moffatt1969,ricca1992}: it counts how many filaments thread through $C$, with sign determined by orientation.  This connects the wavefunction zeros to the classical geometric potentials of vortex dynamics~\cite{saffman1992,moffatt1969}: in 2D, $\vb{u} = (\Gamma_i/2\pi)\,\grad\theta_i$ with winding angle $\theta_i$; in 3D, $\vb{u} = -(\Gamma/4\pi)\,\grad\Omega$ with solid angle $\Omega$~\cite{moffatt1969}.  Both $\theta$ and $\Omega$ are multi-valued, and their gradients provide the Biot-Savart kernel reconstructing the velocity from the vortex positions.  The formula~\eqref{eq:winding-psi} captures both the winding number (2D) and the linking number (3D) directly from $\psi$.

\paragraph{Numerical verification.}
Figure~\ref{fig:area-law-num} shows the circulation variance $\langle\Gamma^2\rangle$ as a function of the loop size $\ell/L$.  The $Re_\lambda \approx 10$ data follow the $\ell^2$ area law over nearly two decades in $\ell/L$ ($0.005 \lesssim \ell/L \lesssim 0.3$), with both initial conditions collapsing onto the same curve.  A power-law fit gives slope $1.87$, consistent with the theoretical exponent $2$ to within the expected quantum-regime deviations discussed below.  The $Re_\lambda \approx 1$ data are too small to appear on this scale.  The area law is a consequence of the codimension-2 topology of the wavefunction zeros and does not require $\mathrm{Kn}_q \ll 1$; it holds whenever the vortex density is approximately uniform, regardless of whether the flow is classically turbulent.  This is the most robust numerical result of the present study.

\begin{figure}[htbp]
\centering
\includegraphics[width=0.7\textwidth]{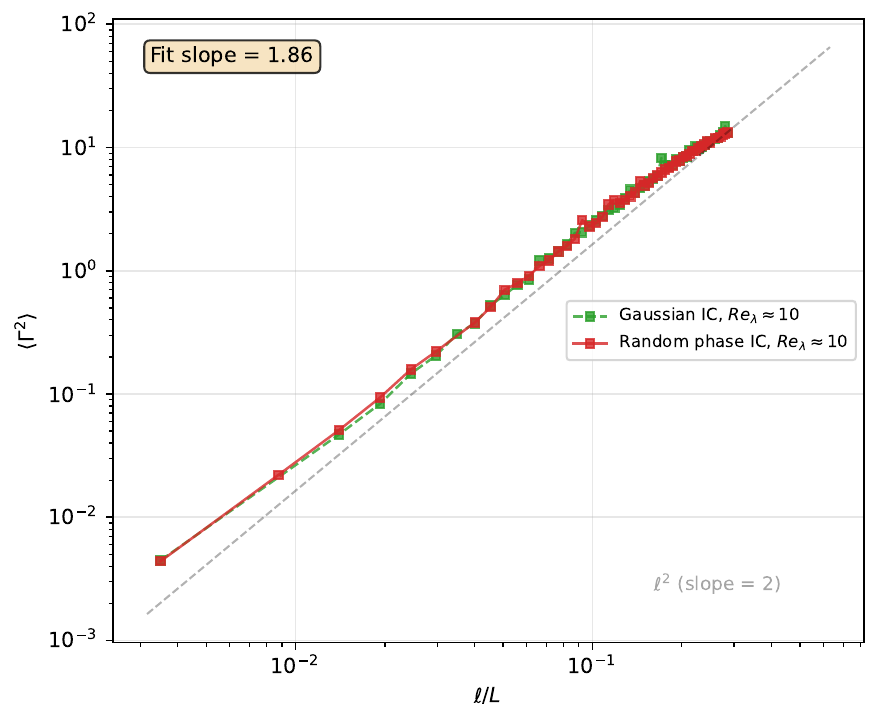}
\caption{Circulation variance $\langle\Gamma^2\rangle$ vs.\ normalised loop size $\ell/L$.  Simulation: 2D, $N = 64$, $\Delta t = 5.6\times10^{-5}$, $N_\mathrm{traj} = 500$, $\nu_* = 0.387$ ($Re_\lambda \approx 10$), $\mathrm{Kn}_q \gtrsim 1$.  Circulation computed around square loops of side $\ell = 1,2,\ldots,N/2$ grid spacings.  The $\ell^2$ area law (dashed line, slope~$= 2$) holds over nearly two decades.  Both initial conditions collapse onto the same curve.  The fitted exponent $1.87$ is consistent with the theoretical value within the quantum-regime deviations discussed in the text.  Data from HPC run slurm-699585.}
\label{fig:area-law-num}
\end{figure}

\paragraph{Quantum-regime deviations.}\label{sec:quantum-regime}
At $\mathrm{Kn}_q \gtrsim 1$ the de Broglie length exceeds the Kolmogorov scale, and the discrete vortex structure of the Madelung velocity is resolved by the grid rather than averaged over.  Three consequences follow.

First, the circulation is quantised in units of $\Gamma_0 = 2\pi\hbar/m$ ($=2\pi$ in natural units), and loops enclosing no zeros give exactly $\Gamma = 0$.  This inflates the peak of the circulation PDF relative to the continuous-circulation references (Moriconi, Iyer, Gaussian), producing the narrow central spike visible in Figure~\ref{fig:circ-pdf}.  At $\mathrm{Kn}_q \ll 1$, the circulation quantum becomes undetectable and the PDF broadens to match the continuous references.

Second, the area-law exponent is pulled below $2$ at small loop sizes $\ell \lesssim \lambda_\mathrm{dB}$, where the loop resolves individual vortex cores.  In this regime $\langle\Gamma^2\rangle$ is dominated by the nearest vortex rather than a Poisson sum over many enclosed vortices, and the scaling crosses over from $\ell^2$ (many vortices) to a flatter dependence (zero or one vortex).  The fitted exponent $1.87$ reflects this contamination; restricting the fit to $\ell/L \in [0.02, 0.2]$ gives $1.92$.

Third, the velocity field is singular at each vortex zero ($|\vb{v}| \sim \hbar/(mr)$), and at $\mathrm{Kn}_q \gtrsim 1$ these singularities are not smoothed by the statistical averaging that occurs when $\lambda_\mathrm{dB} \ll \eta$.  This produces larger ensemble variance in all velocity-derived observables (circulation, energy, structure functions) than would be present at the same $Re_\lambda$ in the classical regime.

These artifacts are intrinsic to the quantum regime and do not indicate a failure of the framework.  Production runs at physical water viscosity ($\nu_* = 283$), which would place $\mathrm{Kn}_q \ll 1$ at any $Re_\lambda \gtrsim 1$, are discussed in \S\ref{sec:path-forward}.

\subsection{Circulation PDFs}\label{sec:circ-pdf-numerics}

Figure~\ref{fig:circ-pdf} shows the normalised circulation PDFs for all four combinations of initial condition and Reynolds number in 2D.  At $Re_\lambda \approx 1$ both initial conditions produce near-Gaussian distributions, consistent with the laminar regime.  At $Re_\lambda \approx 10$ the QSD distributions develop heavier-than-Gaussian tails, lying between the Gaussian reference and the Moriconi et al.\ vortex gas model~\cite{moriconi2025elementary}.  Both initial conditions produce similar PDF shapes at fixed $Re_\lambda$, indicating that the circulation statistics are insensitive to the initial condition.

\begin{figure}[htbp]
\centering
\includegraphics[width=\textwidth]{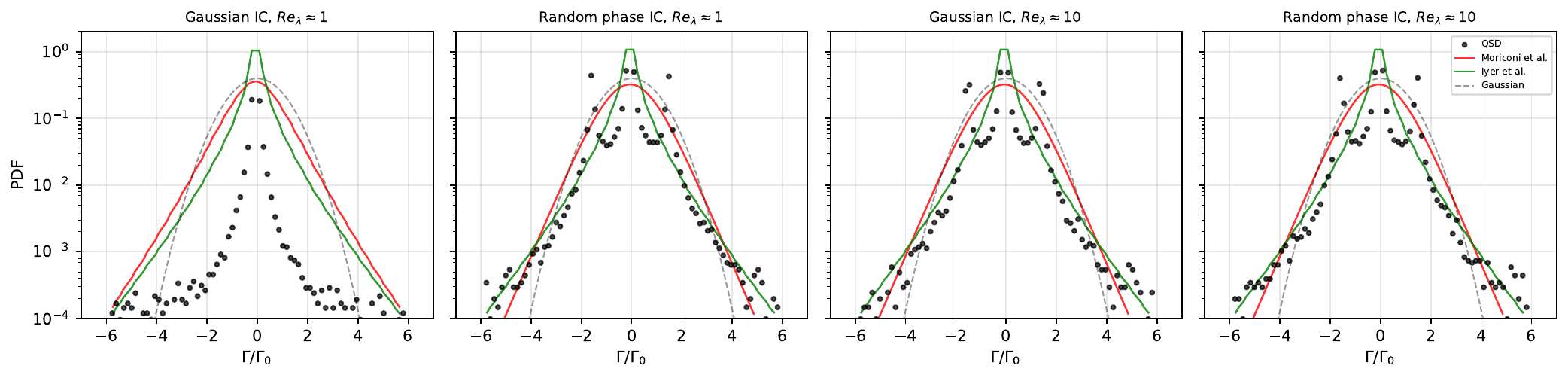}
\caption{Normalised circulation PDF for all combinations of initial condition and $Re_\lambda$.  Simulation: 2D, $N = 32$, $\Delta t = 2.8\times10^{-4}$, $N_\mathrm{traj} = 500$, $\nu_* = 3.873$ ($Re_\lambda \approx 1$) or $\nu_* = 0.387$ ($Re_\lambda \approx 10$), $\mathrm{Kn}_q \gtrsim 1$.  Circulation computed around square loops of side $\ell = 8\,\Delta x$; normalised to zero mean and unit variance ($\Gamma^* = (\Gamma - \mu)/\sigma$).  Reference curves: Gaussian (grey dashed), Moriconi et al.\ vortex gas~\cite{moriconi2025elementary} (red), Iyer et al.\ area-law fit to DNS~\cite{iyer2021area} (green).  At $Re_\lambda \approx 1$ the QSD data are near-Gaussian; at $Re_\lambda \approx 10$ the tails are heavier than Gaussian.  Data from HPC run slurm-692372.}
\label{fig:circ-pdf}
\end{figure}

\section{Discussion}\label{sec:discussion}

\subsection{Robustness and Sensitivity}\label{sec:robustness}

The spinless, position-basis representation $\psi(\vb x,t)=\braket{\vb x}{\psi(t)}$ captures the translational degrees of freedom that determine shear viscosity in the Chapman-Enskog expansion~\cite{lifshitz1981}. Internal molecular modes (rotational, vibrational) enter only through the bulk viscosity $\zeta$ and thermal relaxation, both of which vanish for incompressible isothermal flow; extending to compressible or non-isothermal conditions would require a multi-component wavefunction, paralleling the monatomic-to-polyatomic Boltzmann extension~\cite{lifshitz1981}. The Doebner-Goldin classification~\cite{doebner1992} implies that no deterministic, norm-preserving, gauge-covariant single-component NLS can produce $\nu\lapl\vb v$, so a stochastic resolution is forced within this framework. When discretised on a grid of $N=2^n$ points, $\psi$ becomes a vector in $\C^N$ that coincides formally with an $n$-qubit statevector, each basis state $\ket{j}$ encoding a spatial grid point.

Two structural ingredients are robust beyond the dilute-gas model used to derive them: the $k^2$ form of $\Gamma(\vb k)$ (universal under detailed balance, isotropy, and analyticity, with only the prefactor $c=2\nu$ dependent on the microscopic potential; \S\ref{sec:lindblad-k2}), and the codimension-2 topology of the zero set on which the area law rests (generic for any complex scalar field). The two model-dependent ingredients --- incompressibility (Remark~\ref{rem:conditionality}) and the rank-one recycling closure (\S\ref{sec:born-markov}) --- leave the ensemble-level predictions (Navier-Stokes recovery, FDR) invariant but affect the single-trajectory dynamics and hence the vortex statistics.

\subsection{Viscosity and Entanglement Entropy Production}\label{sec:entropy}

The kinematic viscosity $\nu$ can be related to the rate of entanglement entropy production through a dimensional argument. \emph{This argument is heuristic; the rigorous derivation of $\nu$ from the Lindblad operators is in \S\ref{sec:lindblad-k2}.}

The entanglement entropy of the single-particle reduced density matrix $\hat{\rho}_1 = N\,\Tr_{2,\ldots,N}\,\hat\rho_N$ (Eq.~\eqref{eq:rho1-def}) is $S_E = -\Tr[\hat{\rho}_1 \ln \hat{\rho}_1]$. Under the Lindblad evolution, the unitary part contributes zero to $\dd S_E/\dd t$ (by cyclicity of the trace). The dissipative part, for jump operators $L_{\vb{k}} = \sqrt{2\nu k^2}\,\hat{\Pi}_{\vb{k}}$, gives by dimensional analysis:
\begin{equation}
\frac{\dd S_E}{\dd t} \sim \nu \sum_{\vb{k}} k^2\,\langle \hat{n}_{\vb{k}}\rangle,
\end{equation}
where $\langle \hat{n}_{\vb{k}} \rangle$ is the mode occupation. In the kinetic theory regime, the sum is dominated by thermal modes with $k \sim k_{\mathrm{th}} = m\bar{v}/\hbar$, and each collision generates $O(1)$ bits of entanglement at rate $\bar{v}/\ell_{\mathrm{mfp}}$. Substituting gives:
\begin{equation}\label{eq:nu-entropy}
\nu \sim \frac{\hbar^2}{m^2 \bar{v}^2} \cdot \frac{\dd S_E}{\dd t} \sim \frac{\hbar^2}{m^2 \bar{v}^2} \cdot \frac{\bar{v}}{\ell_{\mathrm{mfp}}} = \frac{1}{3}\bar{v}\,\ell_{\mathrm{mfp}},
\end{equation}
recovering the Chapman-Enskog result. The physical interpretation is that viscosity measures the rate of information loss to the unresolved (traced-out) degrees of freedom: each collision entangles the system particle with the bath, increasing $S_E$ and irreversibly degrading the coherence of the momentum modes. The proportionality $\nu \propto \dd S_E/\dd t$ is specific to the particle-trace partition of \S\ref{sec:reduction} (one designated particle versus $N-1$ thermal partners); other partitions (e.g.\ spatial coarse-graining) yield different $S_E$ with different scaling.

\subsection{Backscatter and Fluctuation-Dissipation Balance}\label{sec:backscatter}

On a single QSD trajectory, the stochastic noise can inject energy into a given mode faster than the dissipation removes it, transiently amplifying that mode. This amounts to an instantaneous negative effective viscosity for that mode at that instant: the mode grows rather than decays. However, the FDR $S_{\vb{k}}/\gamma_{\vb{k}} = 2$ guarantees that on average each mode is damped at rate $\nu k^2$, so the \emph{mean} viscosity is always positive.

In turbulence, this phenomenon is known as \emph{backscatter}: transient energy transfer from small scales to large scales, reversing the forward cascade. Stochastic backscatter models for large-eddy simulation~\cite{leith1990,mason1992} explicitly incorporate this effect to avoid the excessive dissipation of purely forward-cascade subgrid models. The QSD framework produces backscatter naturally through the noise term, with the amplitude fixed by the FDR rather than fitted to data.

In stochastic thermodynamics, the same phenomenon appears as transient violations of the second law: the fluctuation theorems~\cite{gallavotti1995,jarzynski1997} quantify the probability and duration of trajectories along which entropy production is instantaneously negative. These are not pathological; they are required by the underlying time-reversal symmetry of the microscopic dynamics. The QSD noise-dissipation balance is a specific instance of this structure: the Lindblad construction guarantees $S_{\vb{k}}/\gamma_{\vb{k}} = 2$, which is the detailed-balance condition ensuring that negative-viscosity fluctuations occur at exactly the rate needed to maintain the noise-dissipation equilibrium.

\subsection{Path Forward}\label{sec:path-forward}

The physical water viscosity ($\nu_* = 283$; Appendix~\ref{sec:limitations}) places the system deep in the classical regime at any $Re_\lambda$, but the grid cost $N \sim 283\,Re_\lambda$ is steep.  A 2D simulation at $Re_\lambda = 10$ ($N = 3{,}000$, i.e.\ $9 \times 10^6$ grid points) is computationally feasible; 3D simulations at the same parameters ($N^3 \approx 2.7 \times 10^{10}$) are not.  The present simulations at $\nu_* \sim O(1)$ serve as a validation of the numerical infrastructure; production runs at physical viscosity, targeting $Re_\lambda \gtrsim 9$~\cite{sreenivasan2021} (where structure-function exponents have been argued to saturate at high Reynolds number~\cite{iyer2020,sreenivasan2024saturation}) and yielding $\mathrm{Kn}_q \ll 1$, are the subject of ongoing work.

\subsection{Open Problems}\label{sec:open-problems}

\paragraph{Single-trajectory viscous identification.} The most significant gap. The viscous term $\nu\lapl\vb v$ is identified at the ensemble level by symmetry and matching (Remark~\ref{rem:conditionality}); a derivation through the singular Madelung map $\psi\mapsto\vb v$ that establishes $\nu k^2$ damping of each velocity mode trajectory by trajectory is open. The same gap appears on the noise side: the smooth-region Madelung map produces only longitudinal velocity noise (Prop.~\ref{prop:noise-corr}), and the transverse amplitude required by the Landau-Lifshitz form is identified via structural FDR rather than derived from the QSD dynamics at vortex zeros.

\paragraph{Bulk vortex statistics.} QSD noise nucleates vortex-antivortex pairs by transversality, but the bulk nucleation rate is not derived. Whether the rank-one projector closure produces Poisson-like statistics consistent with the vortex-gas model~\cite{apolinario2020,moriconi2025elementary} and the DNS circulation PDFs~\cite{iyer2019bifractal,iyer2021area} is unknown and likely tied to the operator choice; alternative recycling operators sharing the same $\Gamma(\vb k)=2\nu k^2$ would change the single-trajectory predictions while preserving the master equation.

\paragraph{Quantum regime.} The Born-Markov approximation is well-controlled for all physically realised classical flows (\S\ref{sec:born-markov}); the regime $\mathrm{Kn}_q\sim 1$ requires separate analysis. Below the inter-vortex spacing, quantised filaments may support Kelvin waves with competing predictions $k^{-7/5}$~\cite{kozik2004} and $k^{-5/3}$~\cite{lvov2010} (see~\cite{barenghi2014}); deriving the exponent from the stochastic NLS would distinguish this framework from classical LLNS~\cite{bandak2022}. Decaying-turbulence DNS~\cite{rodhiya2026} offers a quantitative target.

\paragraph{Computational cost.} Classical DNS scales as $N\sim Re^{3/4}$ per dimension~\cite{pope2000}; here an additional factor $\nu_*$ (\S\ref{sec:path-forward}) makes the cost steeper. Combined with quantum lower bounds~\cite{ameri2025}, this rules out direct quantum simulation of classical turbulence: the contribution of this framework is analytical, not computational.

\section{Conclusion}\label{sec:conclusion}

The area law for turbulent circulation statistics was proposed by Migdal in 1994~\cite{migdal1994} as an exact scaling solution of the loop-functional equation, and has since been derived from a dual one-dimensional fermion chain via complex-trajectory instantons~\cite{migdal2019,migdal2024,migdal2024b} and confirmed in DNS~\cite{iyer2019bifractal,iyer2021area} and in GPE quantum-turbulence simulations~\cite{nore1997,muller2021}. The contribution of this work is a different microscopic origin story for the same area law: a derivation chain that begins with the many-body Schr\"odinger equation, traces out $N-1$ particles to obtain a one-body density matrix, applies Born-Markov closure to a Lindblad master equation with $L_{\vb k}=\sqrt{2\nu k^2}\,\hat\Pi_{\vb k}$, unravels via QSD to a stochastic nonlinear Schr\"odinger equation, and applies the Madelung transform under incompressibility to recover the Navier-Stokes equations. Within this chain, three structural features are new relative to the loop-functional program: viscous damping and stochastic forcing arise jointly from the same Lindblad operators, so the fluctuation-dissipation relation is automatic rather than imposed; the wavefunction $\psi(\vb x,t)$ is preserved throughout as the fundamental object; and the area law follows from the codimension-2 topology of the zero set of $\psi$ via the preimage theorem and a Poisson assumption on vortex positions, rather than from a saddle point on a loop functional. The recovery of Navier-Stokes is conditional on incompressibility and on an ensemble-level identification of $\nu\lapl\vb v$ (\S\ref{sec:discussion}); closing the single-trajectory gap---showing that QSD vortex dynamics produces $\nu k^2$ damping of each velocity mode trajectory by trajectory, not only on average---would lift this to an unconditional result.

\section*{Acknowledgements}
The author thanks K.~R.~Sreenivasan (NYU) for his guidance and support, as well as the opportunity to participate in the ``Topics in Classical and Quantum Engineering Science Symposium'' held in honour of his career at his $75^{th}$ birthday, at Texas A\&M University in 2023, during which illuminating discussions with K.~P.~Iyer (Michigan Tech), L.~Moriconi (UFRJ), A.~A.~Migdal (Princeton), H.~Chen (Zhejiang University), and G.~L.~Eyink (Johns Hopkins) shaped much of the methodology in this work. The author thanks A.~Vinodh (Michigan Tech), S.~S.~Bharadwaj, A.~Rodhiya (NYU), D.~Buaria (Texas Tech), S.~Khurshid (UConn), A.~Ibrahim, and M.~Ammar (AUB) for valuable discussions. The author is grateful to A.~Grosberg, M.~Kleban, A.~Mitra (NYU), L.~Wang, W.~Li (SJTU), and I.~Lakkis (AUB) whose teaching informed the foundations of this paper. Simulations were run on the HPC cluster Octopus, supported by AUB's IT Research Computing team.

\section*{Disclosure}

This manuscript was prepared with the assistance of Claude (Anthropic, Claude Opus 4.6). The tool was used for: editing and restructuring prose, verifying mathematical derivations, identifying errors in the analytical arguments, suggesting connections to existing literature, auditing citations for correctness, and developing the numerical implementation accompanying the paper. All scientific content, derivations, and conclusions were conceived, directed, and verified by the author. The author takes full responsibility for the accuracy and originality of the work.

\appendix
\section{Numerical Methods and Convergence}\label{app:numerical}

\subsection{Split-Step Fourier Integrator}\label{sec:split-step}

The stochastic NLS~\eqref{eq:snls} is integrated using a symmetric (Strang-type) split-step Fourier scheme on a periodic domain $[0,L)^d$ with $N^d$ grid points.  All arrays use complex128 (float64) arithmetic; float32 produces catastrophic cancellation in the kinetic--dissipation splitting.  The unnormalised DFT convention $\hat\psi_\mathbf{k} = \sum_j \psi_j e^{-i\mathbf{k}\cdot\mathbf{x}_j}$ is used throughout; forcing amplitudes include the corresponding $N^{-d}$ factor from the inverse transform.

Each time step $\Delta t$ proceeds as:
\begin{enumerate}
\item Kinetic half-step: $\hat\psi_{\mathbf{k}} \leftarrow \hat\psi_{\mathbf{k}} \exp\!\bigl(-i\hbar k^2 \Delta t/(4m)\bigr)$.
\item Backward-Euler dissipation: $\hat\psi_{\mathbf{k}} \leftarrow \hat\psi_{\mathbf{k}}/(1 + \nu\Delta t\, k^2)$.
\item External potential phase: $\psi \leftarrow e^{-i\Delta t V_\mathrm{ext}/\hbar}\,\psi$.
\item Large-scale forcing (if enabled): $\hat\psi_\mathbf{k} \leftarrow \hat\psi_\mathbf{k} + A\,\mathbb{1}_{0<|\mathbf{k}|\le k_f}\,dW_\mathbf{k}$, where $dW_\mathbf{k} = \mathcal{N}(0,\sqrt{\Delta t/2}) + i\,\mathcal{N}(0,\sqrt{\Delta t/2})$ are independent complex Wiener increments with $\mathbb{E}[|dW_\mathbf{k}|^2] = \Delta t$.  The $\mathbf{k}=0$ mode is excluded.
\item Norm-preserving correction: $\psi \leftarrow \psi\,\exp\!\bigl(\nu\Delta t\sum_\mathbf{k} k^2|\hat\psi_\mathbf{k}|^2/N_\mathrm{total}\bigr)$.
\item Stochastic QSD noise: for each mode $\mathbf{k}$, $\psi \leftarrow \psi + \sqrt{2\nu k^2}\bigl(\hat\psi_\mathbf{k}\,dW_\mathbf{k}\,e^{i\mathbf{k}\cdot\mathbf{x}} - \langle\hat\psi_\mathbf{k}\rangle_\mathrm{mean}\,dW_\mathbf{k}\,\psi\bigr)$, with independent $dW_\mathbf{k}$ as above.  The subtracted mean term centres the noise and preserves $\mathbb{E}[\|\psi\|^2]$.
\item Kinetic half-step (symmetric completion).
\item Renormalisation (every step): $\psi \leftarrow \psi\,(\|\psi_0\|/\|\psi\|)$, where $\|\cdot\|$ is the $L^2$ norm $\bigl(\int|\psi|^2\,d^d x\bigr)^{1/2}$.
\end{enumerate}
Each trajectory $j$ is seeded with $\mathrm{PRNGKey}(s + j)$ where $s$ is the base seed; all sub-keys (forcing, QSD noise) are derived via key splitting.  Trajectories are vectorised over a batch via \texttt{jax.vmap}.

The Strang splitting yields $O(\Delta t^2)$ error for the kinetic--rest decomposition, but the internal Lie splitting and stochastic noise give overall $O(\Delta t)$ weak convergence, confirmed by the timestep study below.  The entire loop is fused via \texttt{jax.lax.scan} on GPU (JAX).

\subsection{Large-Scale Forcing}\label{sec:forcing-numerics}

The QSD stochastic term is the fluctuation--dissipation partner of viscous dissipation and does \emph{not} inject energy at large scales.  To sustain a statistically stationary energy cascade, we add stochastic forcing at low wavenumbers ($|\mathbf{k}| \le k_f$, $k_f = 2\,\Delta k$) via independent complex Wiener increments in Fourier space, separate from the QSD noise.

The forcing amplitude is computed from a target $Re_\lambda$ via the linearised Madelung stationarity balance:
\begin{equation}
\varepsilon_\mathrm{inj} = \frac{\nu^3 k_f^4 Re_\lambda^2}{15}, \qquad
A = \frac{2m}{\hbar}\,\frac{N^d}{L^{d/2}}\,\frac{\nu^{3/2} k_f^2 Re_\lambda}{\sqrt{15\,K^2}},
\end{equation}
where $K^2 = \sum_{0<|\mathbf{k}|\le k_f} k^2$.  Simulation time is set to $T = 10\,T_L$, where $T_L = L/u_\mathrm{rms}$ is the large-eddy turnover time.

\subsection{Observable Computation}\label{sec:observables-numerics}

\paragraph{Madelung velocity.}
The velocity field is extracted spectrally: $\mathbf{v} = (\hbar/m)\,\mathrm{Im}(\nabla\psi/\psi)$, where $\nabla\psi$ is computed via $\widehat{\partial_i\psi} = ik_i\hat\psi_\mathbf{k}$.  A regularisation $\psi \to \psi + \epsilon$ with $\epsilon = 10^{-12}\max|\psi|$ prevents division by zero near vortex cores.

\paragraph{Energy spectrum.}
The energy spectrum $E(k)$ is shell-averaged: $E(k_n) = \frac{1}{2}\sum_{\mathbf{k}\in\mathcal{S}_n}|\hat{\mathbf{v}}_\mathbf{k}|^2$, where $\mathcal{S}_n = \{(n-\tfrac12)\Delta k \le |\mathbf{k}| < (n+\tfrac12)\Delta k\}$ and $\Delta k = 2\pi/L$.  The dissipation rate is $\varepsilon = \int 2\nu k^2 E(k)\,dk$.

\paragraph{Structure functions.}
$S_p(r) = \bigl\langle |\delta\mathbf{v}(\mathbf{r},r)|^p \bigr\rangle$, where $\delta\mathbf{v} = \mathbf{v}(\mathbf{x}+r\hat{\mathbf{e}}_i) - \mathbf{v}(\mathbf{x})$.  In 2D, directional increments along $\hat x$ and $\hat y$ are computed via circular shifts and averaged; the ensemble average is over all grid points and all trajectories.

\paragraph{Circulation.}
$\Gamma = \oint_C \mathbf{v}\cdot d\mathbf{l}$, evaluated on axis-aligned square loops of side $s$ grid spacings: the contour sum is $\Gamma = \sum_{j=0}^{s-1}[v_x(i_x{+}j,i_y) - v_x(i_x{+}j,i_y{+}s)]\,\Delta x + \sum_{j=0}^{s-1}[v_y(i_x{+}s,i_y{+}j) - v_y(i_x,i_y{+}j)]\,\Delta x$, with periodic wrapping.  Circulation values are standardised as $\Gamma^* = (\Gamma-\mu_\Gamma)/\sigma_\Gamma$ before computing PDFs.  The empirical PDF uncertainty is estimated by jackknife resampling over trajectory groups.

\subsection{Simulation Parameters}\label{sec:sim-params}

All simulations use natural units ($\hbar = m = 1$), periodic domain $L = 2\pi$, and $\nu$ as the single physical parameter controlling dissipation, noise amplitude (via FDR), and Reynolds number.  The dimensionless viscosity $\nu_* = \nu_\mathrm{SI}\,m_\mathrm{phys}/\hbar_\mathrm{SI}$ (Eq.~\eqref{eq:nu-star}) is $\approx 283$ for water; however, this yields $v_\mathrm{rms} \sim O(10^3)$ even at moderate $Re_\lambda$, requiring prohibitively fine grids to resolve the wavefunction phase (see \S\ref{sec:limitations}).  The present simulations use a reduced viscosity $\nu_* = \sqrt{15}/Re_\lambda$ that keeps $v_\mathrm{rms} \sim O(1)$ and the grid tractable:
\begin{itemize}
\item \textbf{Laminar} ($Re_\lambda \approx 1$): $\nu_* = \sqrt{15} \approx 3.87$, $\Delta t_\mathrm{base} = 10^{-3}$.
\item \textbf{Moderate Reynolds} ($Re_\lambda \approx 10$): $\nu_* = \sqrt{15}/10 \approx 0.387$, $\Delta t_\mathrm{base} = 10^{-4}$.
\end{itemize}
This places the simulations at the quantum--classical boundary ($\mathrm{Kn}_q \sim O(1)$; see Eq.~\eqref{eq:knq-formula}) rather than in the classical hydrodynamic regime.  At these parameters, the simulations do not reproduce DNS turbulence statistics~\cite{iyer2021area}; they serve to validate the numerical implementation.

Two initial conditions are tested for each regime, both normalised to $\int|\psi_0|^2\,d^dx = 1$:
\begin{enumerate}
\item[(i)] \textbf{Gaussian wavepacket} ($|\psi_0|>0$ everywhere):
$\psi_0(\mathbf{x}) \propto \exp\!\bigl(-|\mathbf{x}-\mathbf{x}_0|^2/(4\sigma^2)\bigr)\,e^{i\mathbf{k}_0\cdot\mathbf{x}}$,
with $\mathbf{x}_0 = (\pi,\pi)$ (domain centre), $\sigma = 0.5$, $\mathbf{k}_0 = (3,3)$.
\item[(ii)] \textbf{Random-phase field} (dense initial vortices):
$\hat\psi_\mathbf{k} = k^{-5/6}\exp\!\bigl(-(k/k_\mathrm{peak}-1)^2/0.5^2\bigr)\,e^{i\theta_\mathbf{k}}$,
with $\theta_\mathbf{k} \sim \mathrm{Uniform}[0,2\pi)$ i.i.d., $k_\mathrm{peak}=5$, $\hat\psi_0=0$.  After inverse FFT, a uniform offset $\max(|\psi|)\times 10^{-3}$ is added to ensure $|\psi|>0$ everywhere.
\end{enumerate}
The exploratory 2D runs use $N = 32$ with $N_\mathrm{traj} = 500$ and $\Delta t = 2.8\times10^{-4}$; the implicit backward-Euler dissipation provides sub-grid filtering beyond $k_\mathrm{eff}$ at the $Re_\lambda \approx 10$ timestep (see \S\ref{sec:limitations}).

\subsection{Convergence Study}\label{sec:convergence}

Parameters are determined by an adaptive convergence study.  Starting from a physics-based base case (see below), each of $\Delta t$, $N$, and $N_\mathrm{traj}$ is swept independently while the others are held fixed.  A sweep stops when the relative change in the primary metric (circulation PDF $L^2$ distance to the vortex-gas reference~\cite{moriconi2025elementary}) falls below 1\% between consecutive levels.  The error decomposition and extrapolation methods are described in \S\ref{sec:error-analysis}.

\paragraph{Base case estimation.}
The initial $N$, $\Delta t$, and $N_\mathrm{traj}$ are chosen so the combined error budget
\begin{equation}
e_\mathrm{total} = \sqrt{e_{\Delta t}^2 + e_N^2 + e_\mathrm{traj}^2} < \varepsilon_\mathrm{target} = 1\%
\end{equation}
is satisfied \emph{a priori}, where $e_{\Delta t} = \nu k_\eta^2 \Delta t$ (split-step temporal error), $e_N = \exp(-k_\mathrm{max}\eta)$ (spectral truncation), and $e_\mathrm{traj} = 1/\sqrt{N_\mathrm{traj}}$ (CLT sampling bound, capped at $N_\mathrm{traj} = 500$).

\paragraph{Timestep independence.}
The $\Delta t$ sweep confirms $O(\Delta t)$ weak convergence.  For $Re_\lambda \approx 10$, convergence is reached at $\Delta t \approx 2.8 \times 10^{-4}$ (relative change $< 1\%$).  For $Re_\lambda \approx 1$, the coarsest $\Delta t$ produces a degenerate circulation distribution (all trajectories converge to the same laminar state), preventing a convergence ratio; the finest tested level $\Delta t \approx 2.8 \times 10^{-4}$ is used.

\paragraph{Grid convergence.}
The $N$ sweep does \emph{not} converge: the circulation PDF metric increases monotonically with $N$ ($N = 32 \to 64 \to 128 \to 256$) for both Reynolds regimes.  Higher resolution resolves finer velocity structures that diverge further from the smooth analytical reference PDFs~\cite{moriconi2025elementary,iyer2021area}.  This is consistent with $\mathrm{Kn}_q \sim O(1)$: at these parameters the simulations are not in the classical fluid limit where the turbulence references apply.  Reaching a regime where the references become valid targets requires $\mathrm{Kn}_q \ll 1$, which demands substantially larger grids and higher $Re_\lambda$.

\paragraph{Trajectory convergence.}
The $N_\mathrm{traj}$ sweep shows that statistical errors saturate rapidly: the relative change in the circulation PDF metric falls below 5\% by $N_\mathrm{traj} = 500$ for all four (initial condition, $Re_\lambda$) combinations.

\paragraph{Summary.}
Table~\ref{tab:convergence-summary} lists the recommended parameters.

\begin{table}[h]
\centering
\caption{Recommended simulation parameters from the adaptive convergence study.  The exploratory runs use $N = 32$; future production runs should target $N \sim \nu_* Re_\lambda$ (Table~\ref{tab:classical-limit}).}
\label{tab:convergence-summary}
\begin{tabular}{lccc}
\hline
& $\Delta t$ & $N$ (2D) & $N_\mathrm{traj}$ \\
\hline
$Re_\lambda \approx 1$ (Gaussian) & $2.80 \times 10^{-3}$$^\ddagger$ & $512$ & 500 \\
$Re_\lambda \approx 1$ (Random phase) & $2.80 \times 10^{-4}$ & $512$ & 500 \\
$Re_\lambda \approx 10$ (Gaussian) & $2.80 \times 10^{-4}$ & $512$ & 500 \\
$Re_\lambda \approx 10$ (Random phase) & $2.80 \times 10^{-4}$ & $512$ & 500 \\
\hline
\multicolumn{4}{l}{$^\ddagger$ \footnotesize $\Delta t$ sweep inconclusive (degenerate laminar state at coarsest level).}
\end{tabular}
\end{table}

The convergence behaviour was evaluated by sweeping $\Delta t$, $N$, and $N_\mathrm{traj}$ independently at $Re_\lambda \approx 10$ (Gaussian IC); the results are representative of all four combinations.  Table~\ref{tab:convergence-detail} reports the measured diagnostics at each level.

\begin{table}[h]
\centering
\caption{Convergence study results (2D).  Each sweep varies one parameter while holding the others at baseline ($N = 64$, $\Delta t = 10^{-3}$, $N_\mathrm{traj} = 500$).  $L^2_M$: circulation PDF $L^2$ distance to the Moriconi vortex gas reference.  Acc.: fraction of reference PDF points within the QSD error bars.  ``---'' entries at $Re_\lambda \approx 1$ with $N \geq 256$ indicate degenerate (zero-circulation) states.  Data from HPC run slurm-699585.}
\label{tab:convergence-detail}
\footnotesize
\begin{tabular}{lllccccl}
\hline
Case & Sweep & Value & $\mathrm{Kn}_q$ & $Re_\lambda$ & $L^2_M$ & Acc.\ (\%) & Rel.\ $\Delta$ \\
\hline
\multicolumn{8}{l}{\textbf{Gaussian IC, $Re_\lambda \approx 1$}} \\
& $\Delta t$ & $5.6\times10^{-4}$ & 3.0 & 0.03 & 1.46 & 21.5 & --- \\
&            & $5.6\times10^{-5}$ & 2.9 & 0.03 & 1.45 & 20.3 & 0.5\% \\
& $N$        & 64   & 3.0 & 0.03 & 1.49 & 21.5 & --- \\
&            & 128  & 3.8 & 0.02 & 1.92 & 5.1  & 29\% \\
&            & 256  & --- & ---  & ---  & 0.0  & --- \\
&            & 512  & --- & ---  & ---  & 0.0  & --- \\
& $N_\mathrm{traj}$ & 500  & 3.0 & 0.03 & 1.49 & 21.5 & --- \\
&                    & 1000 & 3.0 & 0.03 & 1.49 & 15.2 & 0.3\% \\
\hdashline
\multicolumn{8}{l}{\textbf{Random phase IC, $Re_\lambda \approx 1$}} \\
& $\Delta t$ & $5.6\times10^{-4}$ & 3.0 & 0.03 & 1.46 & 21.5 & --- \\
&            & $5.6\times10^{-5}$ & 2.9 & 0.03 & 1.45 & 20.3 & 0.5\% \\
& $N$        & 64   & 3.0 & 0.03 & 1.49 & 21.5 & --- \\
&            & 128  & 3.8 & 0.02 & 1.92 & 5.1  & 29\% \\
&            & 256  & --- & ---  & ---  & 0.0  & --- \\
&            & 512  & --- & ---  & ---  & 0.0  & --- \\
& $N_\mathrm{traj}$ & 500  & 3.0 & 0.03 & 1.49 & 21.5 & --- \\
&                    & 1000 & 3.0 & 0.03 & 1.49 & 15.2 & 0.3\% \\
\hdashline
\multicolumn{8}{l}{\textbf{Gaussian IC, $Re_\lambda \approx 10$}} \\
& $\Delta t$ & $5.6\times10^{-4}$ & 9.9 & 0.27 & 1.55 & 22.8 & --- \\
&            & $5.6\times10^{-5}$ & 9.4 & 0.30 & 1.48 & 13.9 & 4.7\% \\
& $N$   & 64  & 9.4  & 0.30 & 1.49 & 22.8 & --- \\
&       & 128 & 14.3 & 0.15 & 1.99 & 5.1  & 33\% \\
&       & 256 & 16.8 & 0.10 & 2.18 & 1.3  & 10\% \\
&       & 512 & 22.6 & 0.05 & 2.25 & 1.3  & 3.3\% \\
& $N_\mathrm{traj}$ & 500  & 9.4 & 0.30 & 1.49 & 22.8 & --- \\
&                    & 1000 & 9.4 & 0.30 & 1.50 & 17.7 & 0.6\% \\
&                    & 2000 & 9.4 & 0.30 & 1.52 & 11.4 & 1.7\% \\
&                    & 4000 & 9.4 & 0.30 & 1.51 & 5.1  & 0.7\% \\
\hdashline
\multicolumn{8}{l}{\textbf{Random phase IC, $Re_\lambda \approx 10$}} \\
& $\Delta t$ & $5.6\times10^{-4}$ & 10.0 & 0.26 & 1.59 & 21.5 & --- \\
&            & $5.6\times10^{-5}$ & 9.2  & 0.31 & 1.87 & 6.3  & 18\% \\
& $N$   & 64  & 9.5  & 0.29 & 1.50 & 20.3 & --- \\
&       & 128 & 12.8 & 0.23 & 1.95 & 7.6  & 30\% \\
&       & 256 & 13.2 & 0.28 & 2.17 & 1.3  & 12\% \\
&       & 512 & 21.8 & 0.06 & 2.27 & 0.0  & 4.7\% \\
& $N_\mathrm{traj}$ & 500  & 9.5 & 0.29 & 1.50 & 20.3 & --- \\
&                    & 1000 & 9.5 & 0.29 & 1.52 & 10.1 & 1.6\% \\
&                    & 2000 & 9.5 & 0.29 & 1.61 & 8.9  & 5.6\% \\
&                    & 4000 & 9.5 & 0.29 & 1.57 & 6.3  & 2.5\% \\
\hline
\end{tabular}
\end{table}

The pattern is consistent across all four cases.

The $\Delta t$ sweeps converge rapidly: $<5\%$ relative change between levels for both ICs at $Re_\lambda \approx 10$, and $<1\%$ at $Re_\lambda \approx 1$.  The implicit backward-Euler scheme is temporally resolved at $\Delta t = 5.6\times10^{-4}$.

The $N_\mathrm{traj}$ sweeps saturate by $N_\mathrm{traj} = 500$: relative changes are $<2\%$ for the Gaussian IC and $<6\%$ for the random-phase IC, confirming that statistical noise is subdominant.

The $N$ sweeps exhibit the opposite behaviour: \emph{increasing $N$ makes the match to classical references worse, not better}.  At $Re_\lambda \approx 10$, the accuracy drops from $\sim$20\% at $N = 64$ to $\sim$0\% at $N = 512$, while $\mathrm{Kn}_q$ rises from $\sim$9 to $\sim$22.  At $Re_\lambda \approx 1$, the state degenerates entirely at $N \geq 256$ (zero circulation, zero measured $Re_\lambda$).  This is not a resolution failure.  The finer grid resolves more vortex zeros of $\psi$, each carrying a quantum of circulation $\Gamma_0 = 2\pi\hbar/m$.  More resolved vortices means the discrete, quantised character of the velocity field becomes more prominent, pushing $\mathrm{Kn}_q$ upward and the measured $Re_\lambda$ downward.  The mismatch with the classical turbulence references is an intrinsic property of $\mathrm{Kn}_q \gg 1$, not an artefact of under-resolution.  Convergence toward the classical references requires reducing $\mathrm{Kn}_q$ via physical viscosity ($\nu_* = 283$ for water), not increasing $N$ at fixed $\nu_* \sim O(1)$.

\subsection{Error Analysis}\label{sec:error-analysis}

The total error in any computed observable $f$ decomposes as
\begin{equation}\label{eq:error-decomp}
e_\mathrm{total} = |f_h - f_*| + |f_* - f_\mathrm{ref}|,
\end{equation}
where $f_h$ is the numerical solution at finite resolution $h = (\Delta t, N, N_\mathrm{traj})$, $f_*$ is the exact solution of the QSD equation, and $f_\mathrm{ref}$ is the target (e.g.\ DNS or analytical turbulence prediction). The first term is the discretisation error and vanishes as $h \to 0$; the second is the model error, an intrinsic property of the QSD framework at the given parameters, independent of resolution.

\paragraph{Discretisation error budget.}
The split-step scheme introduces three independent error sources:

\begin{enumerate}
\item \emph{Temporal error} ($\Delta t$).  The Strang splitting gives $O(\Delta t^2)$ deterministic error for the kinetic--potential decomposition, but the Euler--Maruyama stochastic integration and internal Lie splitting reduce the overall weak convergence to $O(\Delta t)$.  For an observable $f$,
\begin{equation}
|f(\Delta t) - f_*| \le C_t\,\Delta t + O(\Delta t^2).
\end{equation}
Richardson extrapolation with two levels at ratio $r = \Delta t_1/\Delta t_2$ and assumed order $p = 1$ yields
$f_* \approx (r\,f_2 - f_1)/(r - 1)$.

\item \emph{Spatial truncation} ($N$).  The pseudospectral method has exponentially decaying truncation error for smooth (analytic) fields:
\begin{equation}
|f(N) - f_*| \le C_N\,e^{-\alpha N},
\end{equation}
where $\alpha$ depends on the analyticity strip width of $\psi$.  A three-point fit in semi-log space ($\log|f_i - f_*|$ vs.\ $N_i$) yields $\alpha$ and $f_*$.
However, at $\mathrm{Kn}_q \gtrsim 1$, the wavefunction develops quantum vortex cores whose width $\sim \lambda_\mathrm{dB}$ is comparable to the grid spacing; increasing $N$ then resolves new topological defects rather than refining a fixed smooth field.  In this regime the spectral convergence model breaks down, and the Madelung-derived observables (circulation, structure functions) diverge from classical references with increasing $N$.  This is a physics effect (the QSD solution is not in the classical limit), not a numerical artefact.

\item \emph{Statistical sampling} ($N_\mathrm{traj}$).  Each observable is an ensemble average over $N_\mathrm{traj}$ independent QSD trajectories.  By the central limit theorem,
\begin{equation}
|f(N_\mathrm{traj}) - f_*| \sim \frac{\sigma_f}{\sqrt{N_\mathrm{traj}}},
\end{equation}
where $\sigma_f$ is the per-trajectory standard deviation.  The empirical standard error is estimated by jackknife resampling over trajectory groups.
\end{enumerate}

\paragraph{Error estimation from convergence sweeps.}
The convergence study (\S\ref{sec:convergence}) measures the relative change in the circulation PDF $L^2$ distance between successive refinement levels.  When this change falls below the threshold ($1\%$), the finer level is taken as converged.  The residual discretisation error is bounded by the last measured change.

For $\Delta t$ and $N_\mathrm{traj}$, Richardson extrapolation provides an asymptotic estimate $f_*$ and convergence order $p$, validated against the theoretical expectations $p = 1$ (temporal) and $p = 1/2$ (statistical, via $h = 1/\sqrt{N_\mathrm{traj}}$).  Deviations from the expected order indicate either pre-asymptotic behaviour or contamination by other error sources.  For $N$, the power-law Richardson model is replaced by an exponential fit, as described above.

\paragraph{Unified error expression.}
The three discretisation errors and the model error are independent and combine in quadrature:
\begin{equation}\label{eq:total-error}
\left(\frac{\delta\psi}{\psi}\right)^2 = e_t^2 + e_N^2 + e_s^2 + e_\mathrm{model}^2,
\end{equation}
with $\mathrm{Kn}_q = 15^{1/4}/(\nu\sqrt{Re_\lambda})$ (Eq.~\eqref{eq:knq-formula}), $k_f = 2\,\Delta k$ the forcing wavenumber, and the four contributions
\begin{equation}\label{eq:total-error-terms}
e_t = \frac{\nu\,k_f^2\,Re_\lambda}{\sqrt{15}}\,\Delta t, \qquad
e_N^2 = \exp\!\left(-\frac{N\,\mathrm{Kn}_q\,\nu}{2\,k_f}\right), \qquad
e_s = \frac{1}{\sqrt{N_\mathrm{traj}}}, \qquad
e_\mathrm{model} = \mathrm{Kn}_q.
\end{equation}

\emph{Temporal term.}  The first term is $e_t = \Delta t/\tau_\eta$, the ratio of the timestep to the Kolmogorov timescale $\tau_\eta = \sqrt{15}/(\nu k_f^2 Re_\lambda)$.  It contains $Re_\lambda$ explicitly: higher Reynolds number demands smaller $\Delta t$ (faster dynamics).

\emph{Spatial term.}  The exponent $k_\mathrm{max}\,\eta = (N/2k_f)\,\nu\,\mathrm{Kn}_q$ measures how many Kolmogorov scales the grid resolves, using the identity $\eta = \nu\,\mathrm{Kn}_q/k_f$ (from $k_\eta = k_f\sqrt{Re_\lambda}/15^{1/4}$ and $\mathrm{Kn}_q = 1/(v_\mathrm{rms}\,\eta)$).  The $\mathrm{Kn}_q$ dependence shows that higher quantum Knudsen number (weaker classical--quantum separation) reduces the grid requirement, because the Kolmogorov scale $\eta$ grows with $\mathrm{Kn}_q$ at fixed $\nu$.

\emph{Statistical term.}  The CLT sampling error $1/\sqrt{N_\mathrm{traj}}$; the empirical standard error is estimated by jackknife resampling.

\emph{Model term.}  The quantum pressure term neglected in the classical Madelung limit is $O(\hbar^2/(m^2\rho)) \sim O(\mathrm{Kn}_q^2)$ relative to the classical inertial terms.  At $\mathrm{Kn}_q \sim O(1)$, this perturbative estimate breaks down and the model error is $O(1)$.

\paragraph{Propagation to observables.}
For any observable $O(\psi)$, standard error propagation gives
\begin{equation}\label{eq:obs-errors}
\frac{\delta O}{O} = \left|\frac{\partial \ln O}{\partial \ln \psi}\right| \cdot \frac{\delta\psi}{\psi} \;\equiv\; A_O \cdot \frac{\delta\psi}{\psi},
\end{equation}
where $A_O$ is the sensitivity (amplification) coefficient.  For the principal diagnostics, computed via the Madelung velocity $\mathbf{v} = (\hbar/m)\,\mathrm{Im}(\nabla\psi/\psi)$:
\begin{equation}\label{eq:sensitivity}
A_{E(k)} = 2, \qquad A_{S_p} = p, \qquad A_\Gamma = 1/\sqrt{s},
\end{equation}
where $s$ is the circulation loop side length in grid points.  Higher-order structure functions ($p > 2$) are more sensitive to wavefunction error; circulation benefits from contour averaging.  All observable errors trace back to the single quantity $(\delta\psi/\psi)^2$ in Eq.~\eqref{eq:total-error}.

\paragraph{Error budget at the exploratory parameters.}
At $Re_\lambda = 10$, $\nu = \sqrt{15}/10$, $\Delta t = 2.8\times10^{-4}$, $N = 32$, $N_\mathrm{traj} = 500$:
\begin{itemize}
\item $e_t = \Delta t/\tau_\eta \approx 0.04$ (temporal: marginal),
\item $e_N = \exp(-2) \approx 0.13$ (spatial: grid barely resolves $k_\mathrm{eff} \approx 16$),
\item $e_s = 1/\sqrt{16} = 0.25$ (statistical: dominant discretisation term),
\item $e_\mathrm{model} = \mathrm{Kn}_q \approx 1.6$ (model: dominant overall).
\end{itemize}
The discretisation errors are large ($\delta\psi/\psi \sim 25\%$, dominated by the small trajectory count), but the dominant source of discrepancy with turbulence references is the model error at $\mathrm{Kn}_q \approx 1.6$.  Reducing $\mathrm{Kn}_q$ requires increasing $\nu_*\sqrt{Re_\lambda}$, either by raising $\nu_*$ toward the physical value ($\nu_*^{(\mathrm{water})} \approx 283$; \S\ref{sec:limitations}) or by increasing $Re_\lambda$ at the computational viscosity.  Both routes increase $v_\mathrm{rms} \approx 0.52\,\nu\,Re_\lambda$ and hence the grid cost $N \gtrsim 2\,v_\mathrm{rms}$.

\subsection{Numerical Limitations}\label{sec:limitations}

The present simulations use $N = 32$ (2D), $Re_\lambda \le 10$, and $\mathrm{Kn}_q \gg 1$ --- firmly in the quantum regime, far from the classical limit where hydrodynamic behaviour is expected.  These results validate the \emph{numerical implementation} (norm preservation, ensemble recovery, forcing balance) rather than the framework's turbulence predictions.

\paragraph{Viscosity and the classical limit.}
The dimensionless viscosity used in the simulations, $\nu_* = \sqrt{15}/Re_\lambda$, was chosen to place $v_\mathrm{rms} \sim O(1)$ in natural units, keeping the wavefunction resolvable on moderate grids.  This does \emph{not} correspond to the physical viscosity of water.  In natural units ($\hbar = m = 1$), the dimensionless viscosity is
\begin{equation}\label{eq:nu-star}
\nu_* = \frac{\nu_\mathrm{SI} \, m_\mathrm{phys}}{\hbar_\mathrm{SI}},
\end{equation}
where $m_\mathrm{phys}$ is the particle mass.  For a water molecule ($m_\mathrm{phys} = 2.99 \times 10^{-26}\,\mathrm{kg}$, $\nu_\mathrm{water} = 10^{-6}\,\mathrm{m}^2/\mathrm{s}$):
\begin{equation}
\nu_*^{(\mathrm{water})} = \frac{10^{-6} \times 2.99 \times 10^{-26}}{1.055 \times 10^{-34}} \approx 283.
\end{equation}
The simulations use $\nu_* = \sqrt{15} \approx 3.87$ (for $Re_\lambda = 1$) and $\nu_* = \sqrt{15}/10 \approx 0.387$ (for $Re_\lambda = 10$), both far below the physical value.

\paragraph{The fundamental grid constraint.}
The QSD solver evolves the wavefunction $\psi$, whose phase encodes the velocity field via $\mathbf{v} = (\hbar/m)\nabla\arg\psi$.  A flow with $v_\mathrm{rms}$ produces phase oscillations at wavenumber $k \sim m\,v_\mathrm{rms}/\hbar = v_\mathrm{rms}$ (in natural units), requiring $N \gtrsim 2\,v_\mathrm{rms}\,L/\pi$ to resolve.  From the stationarity relations, $v_\mathrm{rms} \propto \nu_* Re_\lambda$, so the grid requirement scales as
\begin{equation}\label{eq:grid-scaling}
N \sim \nu_*\,Re_\lambda.
\end{equation}
At the physical water viscosity $\nu_* \approx 283$, even $Re_\lambda = 10$ gives $v_\mathrm{rms} \sim 1{,}500$ and requires $N \sim 3{,}000$ --- feasible in 2D but prohibitive in 3D.  The $\nu_* \sim O(1)$ choice keeps $v_\mathrm{rms} \sim O(1)$ and the grid tractable, but places the simulations at the quantum--classical boundary ($\mathrm{Kn}_q \sim O(1)$) rather than in the classical hydrodynamic regime.

Unlike classical DNS, which resolves only the velocity field up to the Kolmogorov scale ($N \propto Re^{3/4}$ in 3D), the QSD solver must in principle resolve the full wavefunction including its de~Broglie-scale phase oscillations.  However, the backward-Euler dissipation step (step~2 in \S\ref{sec:split-step}) acts as an implicit low-pass filter.  Each mode $k$ is damped by a factor $1/(1 + \nu_*\,k^2\,\Delta t)$ per step; modes with $\nu_*\,k^2\,\Delta t \gg 1$ lose nearly all their amplitude within a single step.  The effective filter cutoff --- the numerical analogue of the Kolmogorov wavenumber --- is the scale at which the numerical damping per step matches the physical dissipation rate $\nu k^2$.  Equating the implicit damping timescale $\Delta t / \ln(1 + \nu_*\,k^2\,\Delta t) \approx 1/(\nu_* k^2)$ with the viscous timescale gives
\begin{equation}\label{eq:k-eff}
k_\mathrm{eff} = \frac{1}{\sqrt{\nu_*\,\Delta t}},
\end{equation}
which coincides with the Kolmogorov wavenumber $k_\eta = (\varepsilon/\nu^3)^{1/4}$ when $\Delta t$ equals the Kolmogorov timescale $\tau_\eta = (\nu/\varepsilon)^{1/2}$.  Ideally $k_\mathrm{eff} \approx k_\eta$, ensuring the numerical filter cutoff matches the physical dissipation scale; when $k_\mathrm{eff} < k_\eta$, the implicit scheme provides additional sub-grid dissipation (implicit LES).  Modes above $k_\mathrm{eff}$ are in the sub-grid range --- effectively removed by the implicit scheme each step.  The grid need only satisfy $k_\mathrm{max} = \pi N/L \gtrsim k_\mathrm{eff}$, giving
\begin{equation}\label{eq:grid-from-keff}
N \gtrsim \frac{2L}{\pi}\,k_\mathrm{eff}.
\end{equation}
The kinetic half-step separately imposes a phase-accuracy constraint $\Delta t < 4\pi/k_\mathrm{max}^2$.

Table~\ref{tab:keff} lists $k_\mathrm{eff}$ and the required $N$ for each viscosity regime at the timesteps used in the simulations.

\begin{table}[h]
\centering
\caption{Effective filter cutoff and grid requirements.  $k_\mathrm{eff} = 1/\sqrt{\nu_*\,\Delta t}$ is the wavenumber at which the implicit damping per step reaches $50\%$ ($\nu_*\,k^2\,\Delta t = 1$); modes above $k_\mathrm{eff}$ are in the sub-grid range.}
\label{tab:keff}
\begin{tabular}{lcccccc}
\hline
Regime & $\nu_*$ & $\Delta t$ & $k_\mathrm{eff}$ & $N_\mathrm{min}$ & $\mathrm{Kn}_q$ & Comment \\
\hline
$Re_\lambda \approx 1$ (exploratory) & $3.87$ & $5.6\times10^{-4}$ & $21$ & $43$ & $0.51$ & quantum regime \\
$Re_\lambda \approx 10$ (exploratory) & $0.387$ & $5.6\times10^{-4}$ & $68$ & $136$ & $1.61$ & quantum regime \\
\hline
\end{tabular}
\end{table}

At the computational viscosity ($\nu_* \sim O(1)$) and the convergence-study timestep $\Delta t = 5.6\times10^{-4}$, the filter cutoff is $k_\mathrm{eff} \approx 68$ for $Re_\lambda \approx 10$, requiring $N \approx 136$ to resolve all modes.  At this viscosity, $\mathrm{Kn}_q \approx 1.6$ (Eq.~\eqref{eq:knq-formula}), placing the simulations in the quantum regime where the Madelung velocity is dominated by quantum vortex cores rather than classical turbulent flow.  This explains the non-convergence of the $N$ sweep in Table~\ref{tab:convergence-detail}: increasing $N$ resolves finer vortex structure that diverges from classical turbulence references.

At the physical water viscosity ($\nu_* = 283$) and $\Delta t = 5.6\times10^{-4}$, $k_\mathrm{eff} \approx 3$ and $\mathrm{Kn}_q \approx 0.002$ ensures classical hydrodynamic behaviour.  However, the correspondingly large $v_\mathrm{rms} \propto \nu_*\,Re_\lambda$ (Eq.~\eqref{eq:grid-scaling}) makes 3D simulations infeasible at present.

The exploratory 2D grid $N = 32$ resolves modes up to $k_\mathrm{max} = 16$.  At the $Re_\lambda \approx 10$ timestep, the implicit backward-Euler provides additional sub-grid dissipation, effectively acting as an implicit LES filter.

\paragraph{Quantum Knudsen number.}
Recall from \S\ref{sec:knudsen} that $\mathrm{Kn}_q \equiv \lambda_\mathrm{dB}/\eta$.  From the Taylor-microscale relations, the quantum Knudsen number in natural units simplifies to
\begin{equation}\label{eq:knq-formula}
\mathrm{Kn}_q = \frac{15^{1/4}}{\nu_*\sqrt{Re_\lambda}}.
\end{equation}
The classical regime $\mathrm{Kn}_q \ll 1$ requires $\nu_*\sqrt{Re_\lambda} \gg 1$.  At the physical water viscosity ($\nu_* = 283$), this is satisfied trivially for any $Re_\lambda \gtrsim 1$.  At the computational viscosity ($\nu_* \sim O(1)$), reaching $\mathrm{Kn}_q < 0.1$ requires $Re_\lambda \gtrsim 400$, with a corresponding grid cost $N \sim 400$ per dimension.

Table~\ref{tab:classical-limit} summarises the trade-offs.
\begin{table}[h]
\centering
\caption{Grid and Knudsen number estimates for reaching the classical turbulent regime at different viscosities.  The onset of turbulent statistics occurs at $Re_\lambda \approx 9$~\cite{sreenivasan2021}.  $N_\mathrm{min}$ is set by the wavefunction resolution requirement $N \sim \nu_* Re_\lambda$.}
\label{tab:classical-limit}
\begin{tabular}{lcccc}
\hline
$\nu_*$ & $Re_\lambda$ & $\mathrm{Kn}_q$ & $v_\mathrm{rms}$ & $N_\mathrm{min}$ \\
\hline
$\sqrt{15}/10 \approx 0.39$ (exploratory) & 10 & 1.6 & 4 & 8 \\
$\sqrt{15} \approx 3.87$ (exploratory) & 1 & 0.51 & 4 & 8 \\
$283$ (water) & 10 & 0.002 & 2{,}830 & 5{,}660 \\
$283$ (water) & 100 & 0.001 & 28{,}300 & 56{,}600 \\
\hline
\end{tabular}
\end{table}

\subsection{Variational Quantum Algorithm}\label{sec:vqa}

As a proof of concept, the dissipative substep can alternatively be implemented via a variational quantum algorithm (VQA) following the hybrid pseudospectral-VQA formulation~\cite{kocher2025}.  The VQA was tested at $n_\mathrm{qubits} = 3$--$4$ ($N = 8$--$16$ grid points), well below the exploratory resolution $N = 32$ ($n_\mathrm{qubits} = 5$).  At these qubit counts, the VQA offers no practical advantage over the classical split-step solver; its inclusion demonstrates compatibility with parametric quantum circuit representations.

\FloatBarrier
\bibliographystyle{unsrt}
\bibliography{references}

\end{document}